\documentclass[12pt]{amsart}
\usepackage[utf8]{inputenc}
\usepackage[english]{babel}
\usepackage{amsmath, amssymb, amsthm, amscd,color,comment}
\usepackage{cancel}
\usepackage{cite}
\usepackage{alltt}
\usepackage[dvipsnames]{xcolor}
\usepackage{array}
\usepackage[small,bf,labelsep=period]{caption}
\usepackage{mathtools}
\usepackage{hyperref}

\usepackage{enumerate}
\newtheorem{theorem}{Theorem}[section]
\newtheorem{definition}[theorem]{Definition}
\newtheorem{example}[theorem]{Example}

\newtheorem{remark}[theorem]{Remark}
\newtheorem{lemma}[theorem]{Lemma}
\newtheorem{corollary}[theorem]{Corollary}
\newtheorem{proposition}[theorem]{Proposition}

\DeclarePairedDelimiter\ceil{\lceil}{\rceil}
\DeclarePairedDelimiter\floor{\lfloor}{\rfloor}

\def\la{\lambda}

\def\b{\beta}
\def\si{\sigma}

\def\w{{\bf w}}
\def\x{{\bf x}}
\def\y{{\bf y}}
\def\z{{\bf z}}
\def\Q{{\bf Q}}

\def\P{\mathbf P}
\def\N{\mathbb N}
\def\R{\mathbb R}
\def\Z{\mathbb Z}
\def\cC{\mathcal C}

\def\cL{\mathcal L}

\def\cX{\mathcal X}

\def\fqs{\mathbb F_{q^2}}

\def\fq{\mathbb F_q}

\def\supp{{\rm Supp}}

\def\Div{{\rm Div}}

\def\deg{{\rm deg}}

\def\lub{{\rm lub}}
\def\glb{{\rm glb}}
\def\Id{{\rm Id}}
\def\supp{{\rm Supp}}
\def\res{{\rm res}}

\def\char{\mbox{\rm Char}}

\def\negalpha{\text{\boldmath$\alpha$}}
\def\neglambda{\text{\boldmath$\lambda$}}

\def\neg1{\text{\boldmath$1$}}
\def\nege{\text{\boldmath$e$}}
\def\negbeta{\text{\boldmath$\beta$}}
\def\neggamma{\text{\boldmath$\gamma$}}
\def\negeta{\text{\boldmath$\eta$}}

\def\neggamma{\text{\boldmath$\gamma$}}

\def\negeta{\text{\boldmath$\eta$}}

\def\neg1{\text{\boldmath$1$}}
\def\negi{\text{\boldmath$i$}}

\newcommand{\al}{\alpha}
\newcommand{\be}{\beta}

\begin{document}

\title[The Set of Pure Gaps at Several Rational Places in Function Fields]{The Set of Pure Gaps at Several Rational Places in Function Fields}



\thanks{The first author was partially supported by FAPEMIG. The second author was partially supported by FAPESP Grant 2022/16369-2. The third author was partially funded by CNPq.}

\author{Alonso S. Castellanos, Erik A. R. Mendoza, and Guilherme Tizziotti}

\address{Faculdade de Matemática, Universidade Federal de Uberlândia, Campus Santa Mônica, CEP 38400-902, Uberlândia, Brazil}
\email{alonso.castellanos@ufu.br}

\address{Universidade Estadual de Campinas, Instituto de Matemática, Estatística e Computação Científica, CEP 13083-859, Campinas, Brazil}
\email{erikrm@ime.unicamp.br}

\address{Faculdade de Matemática, Universidade Federal de Uberlândia, Campus Santa Mônica, CEP 38400-902, Uberlândia, Brazil}
\email{guilhermect@ufu.br}

\begin{abstract}
In this work, using maximal elements in generalized Weierstrass semigroups and its relationship with pure gaps, we extend the results in \cite{CMT2024} and provide a way to completely determine the set of pure gaps at several rational places in an arbitrary function field $F$ over a finite field and its cardinality. As an example, we determine the cardinality and a simple explicit description of the set of pure gaps at several rational places distinct to the infinity place on Kummer extensions, which is a different characterization from that presented by Hu and Yang in \cite{HY2018}. Furthermore, we present some applications in coding theory and AG codes with good parameters.
\end{abstract}

\maketitle

\section{Introduction}

Weierstrass semigroups $H(\Q)$ at an $n$-tuple of distinct rational places $\Q = (Q_1,\ldots, Q_n)$ in a function field $F$ and its set of gaps $G(\Q)$ are classic objects of research in algebraic geometry with important applications in coding theory. The concept of pure gap in Weierstrass semigroups was introduced by Homma and Kim \cite{HK2001} in a study over Weierstrass semigroups at a pair of places, and it was extended for several places by Carvalho and Torres in \cite{CT2005}. Pure gaps are very useful in the construction of algebraic-geometry codes and are related to improvements on the Goppa bound of the minimum distance of these codes, see e.g. \cite{CT2005}, \cite{CB2020}, \cite{CT2016}, \cite{HK2001}, \cite{HY2018}, \cite{GM2001}. The set of pure gaps is denoted by $G_0(\Q)$. In addition to applications in coding theory, calculating the cardinality of the sets $G(\Q)$ and $G_0(\Q)$, and explicitly determining its elements are two interesting problems that have attracted the attention of researchers in recent decades. For the case $\Q=(Q_1, Q_2)$, in \cite{K1994}, Kim obtained lower and upper bounds on the cardinality of $G(\Q)$, and Homma, in \cite{H1996}, got an exact expression for such cardinality. Matthews in \cite{GM2001} determined the elements and cardinality of the gap set at two places on the Hermitian function field. Bartoli, Quoos, and Zini \cite{BQZ2018} counted the gaps at two places on certain Kummer extensions. In  \cite{YH2018}, Hu and Yang determined the numbers of gaps and pure gaps at two places on a quotient of the Hermitian function field. Still in the case of two places, in \cite{CMT2024}, the authors provided a way to completely determine the set of pure gaps $G_0(\Q)$ at two places in an arbitrary function field and its cardinality. As an example, in the same work, they completely determine the set of pure gaps and its cardinality for two families of function fields: the $GK$ function field and Kummer extensions. For the case of several places $\Q=(Q_1, \ldots, Q_n)$, with $n > 2$, in \cite{BQZ2018} the authors gave a criterion to find pure gaps at many places and presented families of pure gaps for some function fields. In \cite{HY2018} and \cite{YH2017} Hu and Yang explicitly determined the Weierstrass semigroups and the set of pure gaps at many places on Kummer extensions, and in \cite{YH2018} they gave an arithmetic characterization of the set of pure gaps at several places on a quotient of the Hermitian function field. 

Generalized Weierstrass semigroups, denoted by $\widehat{H}(\Q)$, was a generalization of Weierstrass semigroup proposed by Delgado in \cite{D1990} working with algebraically closed fields. In \cite{BT2006}, Beelen and Tuta\c{s} studied such object over finite fields. On the other hand, important objects in the study of generalized Weierstrass semigroups $\widehat{H}(\Q)$ are the absolute and relative maximal elements (see Definition \ref{defi maximals}). Studies on $\widehat{H}(\Q)$ and its maximal elements can be found in \cite{MTT2019} and \cite{TT2019}.

In this work, using the notion of maximal elements in generalized Weierstrass semigroups and its relationship with pure gaps given by Ten\'orio and Tizziotti in \cite{TT2019}, we provide a way to completely determine the set of pure gaps $G_0(\Q)$ at several rational places in an arbitrary function field $F$ over a finite field and its cardinality. Our results generalize those presented in \cite{CMT2024}. As an example, we determine the cardinality and a simple explicit description of the set of pure gaps $G_0(\Q)$ at several rational places distinct to the infinity place on Kummer extensions, which is a different characterization from that presented by Hu and Yang in \cite{HY2018}. For this we completely determine the sets of absolute and relative maximal elements of the generalized Weierstrass semigroup $\widehat{H}(\Q)$, where each $Q_i$ is distinct to the infinity place. As an application in coding theory, we use our results to construct differential AG codes in many places over Kummer extensions (see Examples \ref{Example1} and \ref{Example2}) which have better relative parameters than the AG codes obtained in \cite[Examples 1 and 3]{HY2018}. In addition, these codes improve the parameters given in MinT’s Tables \cite{MinT}. 

Here is an outline of the paper. In Section $2$, we present some basic results and terminologies concerning Weierstrass and generalized Weierstrass semigroups, and AG codes. In Section $3$ we study the set of absolute and relative maximal elements. In Section $4$, we present a description of the pure gap set in terms of relative maximal elements in the Weierstrass semigroup and an explicit expression for its cardinality. Finally, in Section 5, we completely determine the set of absolute and relative maximal elements of the generalized Weierstrass semigroup and, using the results of the previous section, we determine the cardinality and a simple explicit description of the pure gap set at totally ramified places on Kummer extensions. Furthermore, we construct multi-point AG codes with good parameters.

\section{Preliminaries and Notation}

Throughout this article, we let $q$ be the power of a prime $p$, $\fq$ be the finite field with $q$ elements, and $K$ be the algebraic closure of $\fq$. For $a\in \R$, we denote by $\floor*{a}$ and $\ceil*{a}$ the floor and ceiling function of $a$, respectively. We also denote $\mathbb{N}_{0} = \mathbb{N} \cup \{0\}$, where $\mathbb{N}$ is the set of positive integers.

Given a function field of one variable $F/K$ of genus $g=g(F)$, we denote by $\mathcal P_F$ the set of places in $F$, by $\Div (F)$ the group of divisors in $F$, and for a function $z \in F$ we let $(z)$ and $(z)_\infty$ stand for the principal and pole divisor of $z$, respectively. 

For a divisor $D\in \Div(F)$, the Riemann-Roch space associated to the divisor $D$ is defined by
$$
\cL(D):=\{z\in F: (z)+D\geq 0\}\cup \{0\},
$$
and we denote by $\ell(D)$ its dimension as vector space over $K$.

\subsection{Generalized Weierstrass semigroup} Let $\Q=(Q_1, \dots, Q_n)$ be an $n$-tuple of distinct rational places in $F$. The Weierstrass semigroup $H(\Q)$ and the generalized Weierstrass semigroup $\widehat{H}(\Q)$ of $F$ at $\Q$ are defined, respectively, by the sets
$$
H(\mathbf{Q}) := \left\{(a_{1}, \ldots, a_{n}) \in \mathbb{N}_{0}^ {n} \mbox{ : } \exists \ f \in F \mbox{ with } (f)_{\infty} = \textstyle\sum_{i=1}^ {n} a_{i}Q_{i} \right\}
$$
and
$$
\widehat{H}(\Q):=\{(-v_{Q_1}(f),\dots ,-v_{Q_n}(f))\in \Z^n : f \in R_\Q\setminus\{0\}\},
$$
where $R_\Q$ denotes the ring of functions on $F$ that are regular outside the points in the set
$\{Q_1,\dots, Q_n\}$ and $v_{Q_i}$ is the discrete valuation associated to place $Q_i$. Provided that $q \geq n$, we have that the Weierstrass semigroup of $F$ at $\Q$ can be obtained by the relation $H(\Q)=\widehat{H}(\Q)\cap \N_0^n$, see \cite[Proposition 2]{BT2006}. The elements in the finite complement ${G(\mathbf{Q}):=\N^n_0\setminus H(\mathbf{Q})}$ are called \emph{gaps} of $F$ at $\mathbf{Q}$. A \emph{pure gap} of $F$ at $\Q$ is an $n$-tuple $\negalpha = (\alpha_1, \ldots, \alpha_n) \in G(\mathbf{Q})$ such that $\ell(D_\negalpha(\Q)) = \ell(D_\negalpha(\Q) - Q_j)$ for any $j=1,\ldots,n$, where $D_\negalpha(\Q) = \alpha_1 Q_1 + \cdots + \alpha_n Q_n$. The set of pure gaps of $F$ at $\mathbf{Q}$ will be denoted by $G_0(\mathbf{Q})$.

%\subsection{Maximals elements in generalized Weierstrass semigroup}
An important notion in the study of $\widehat{H}(\Q)$ due to Delgado \cite{D1990} is that of maximal elements (see Definition \ref{defi maximals}). As we will see in what follows, these elements will play an important role in the description of the elements of $G_0(\Q)$. 

Set $I:=\{1,\ldots,n\}$. For $i\in I$, a nonempty subset $J\subsetneq I$, and $\negalpha=(\alpha_1,\ldots,\alpha_n)\in \Z^n$, we shall denote
\begin{itemize}
\item [$\bullet$] $\nabla^n_i(\negalpha):=\{\negbeta\in \widehat{H}(\Q) : \be_i = \al_i \text{ and }\be_j \leq \al_j \text{ for }j\neq i\};$
\item [$\bullet$] $\overline{\nabla}_J (\negalpha):=\{\negbeta\in \Z^n : \be_j=\al_j \text{ for }j\in J \text{ and } \be_i < \al_i \text{ for }i\notin J\};$
\item [$\bullet$] $\nabla_J(\negalpha):=\overline{\nabla}_J(\negalpha)\cap \widehat{H}(\Q);$
\item [$\bullet$] $\overline{\nabla}(\negalpha):=\cup_{i=1}^{n}\overline{\nabla}_i(\negalpha)$, where $\overline{\nabla}_i(\negalpha):= \overline{\nabla}_{\{i\}}(\negalpha);$
\item [$\bullet$] $\nabla(\negalpha):=\overline{\nabla}(\negalpha)\cap \widehat{H}(\Q).$
%\item [$\triangleright$] $\neg1_J$ denotes the $n$-tuple whose the $j$-th coordinate $1$ is if $j\in J$ and $0$ otherwise,
%\item [$\triangleright$] $\neg1:=(1, 1, \dots, 1)$, and 
%\item [$\triangleright$] $D_\negalpha$ will denote the divisor $\al_1 Q_1+\dots+\al_n Q_n$ on $F$.
\end{itemize}

\begin{definition} \label{defi maximals}
An element $\negalpha\in \widehat{H}(\Q)$ is called maximal if $\nabla(\negalpha)=\emptyset$. If moreover $\nabla_J(\negalpha)=\emptyset$ for every $J\subsetneq I$ with $|J| \geq 2$, we say that $\negalpha$ is absolute maximal. If $\negalpha$ is maximal and $\nabla_J(\negalpha)\neq\emptyset$ for every $J\subsetneq I$ with $|J| \geq 2$, we say that $\negalpha\in \widehat{H}(\Q)$ is relative maximal. The sets of absolute and relative maximal elements in $\widehat{H}(\Q)$ will be denoted, respectively, by $\widehat{\Gamma}(\Q)$ and $\widehat{\Lambda}(\Q)$.
\end{definition}

Observe that the notions of absolute and relative maximality coincide when $n=2$.

For $\negalpha = (\alpha_1, \ldots, \alpha_n) \in \Z^n$, we will denote $D_{\negalpha}(\Q) := \alpha_1 Q_1 + \cdots + \alpha_n Q_n$. Furthermore, we denote by $\nege_j$ the $n$-tuple whose $j$-th coordinate is $1$ and the others are $0$, and by $\neg1$ the $n$-tuple whose coordinates are all $1$.

Below we present some results that provide important equivalences involving the concepts of absolute and relative maximal elements.

\begin{proposition}\cite[Proposition 3.2]{MTT2019} \label{absmax} 
Let $\negalpha\in \widehat{H}(\mathbf{Q})$ and assume that $q \geq n$. The following statements are equivalent:
\begin{enumerate}[\rm (i)]
\item $\negalpha\in \widehat{\Gamma}(\mathbf{Q})$;
\item $\nabla_i^n(\negalpha)=\{\negalpha\}$ for all $i\in I$;
\item  $\nabla_i^n(\negalpha)=\{\negalpha\}$ for some $i\in I$;
\item $\ell(D_{\negalpha}(\Q))=\ell(D_{\negalpha}(\Q) - \sum_{i=1}^n Q_i)+1$.
\end{enumerate}
\end{proposition}

\begin{proposition}\cite[Proposition 2.2]{MTT2019} \label{prop 2.1.3 da tese}
Let $\negalpha \in \mathbb{Z}^n$ and assume that $q \geq n$. Then
\begin{enumerate}[\rm (i)]
\item $\negalpha \in \widehat{H} (\mathbf{Q})$ if and only if $\ell(D_{\negalpha}(\Q)) = \ell(D_{\negalpha}(\Q) - Q_i)+1$, for all $i \in I$;
\item $\nabla_{i}^{n}(\negalpha) = \emptyset$ if and only if $\ell(D_{\negalpha}(\Q))=\ell (D_{\negalpha}(\Q) - Q_i)$, for every $i\in I$.
\end{enumerate}
\end{proposition}

The following concept, introduced by Duursma and Park \cite{DP2012}, is important for characterizing maximal elements.

\begin{definition}
Let $P_1$ and $P_2$ be distinct rational places on $F$. A divisor $A\in \Div(F)$ is called a discrepancy with respect to $P_1$ and $P_2$ if $\cL(A)\neq \cL(A-P_2)$ and $\cL(A-P_1)=\cL(A-P_1-P_2)$.
\end{definition}

The next technical lemma will be useful in the computations with discrepancies.

\begin{lemma}\cite[Noether’s Reduction Lemma]{F1969} \label{lemma noether}
Let $D\in \Div(F)$, $P$ be a place on $F$, and $W$ be a canonical divisor on $F$. If $\ell(D)>0$ and $\ell(W-D-P) \neq \ell(W-D)$, then $\ell(D+P)=\ell(D)$.
\end{lemma}

The following results establishes equivalences using the notion of discrepancy.

\begin{proposition}\cite[Proposition 3]{TT2019FF} \label{prop equiv discrepancia}
Let $\negalpha= (\alpha_1, \ldots , \alpha_{n}) \in \mathbb{Z}^n$ and assume that $q\geq n$. The following statements are equivalent:
\begin{enumerate}[\rm (i)]
\item $\negalpha \in \widehat{\Gamma}(\mathbf{Q})$;
\item $D_{\negalpha}(\mathbf{Q})$ is a discrepancy with respect to any pair of distinct places in $\{Q_1,\ldots,Q_{n}\}$.
\end{enumerate}
\end{proposition}

\begin{proposition}\cite[Proposition 2.8]{TT2019} \label{relmax}
Let $\negalpha\in \mathbb{Z}^n$ and assume that $q \geq n$. The following statements are equivalent:
\begin{enumerate}[\rm (i)]
\item $\negalpha\in \widehat{\Lambda}(\Q)$;
%\item {\color{red} $\ell(\negalpha-\negei)=\ell(\negalpha-\negei-\negej)+1$ for every $i,j\in \{1,\ldots,m\}$ with $i\neq j$;}
\item $\nabla(\negalpha)=\emptyset$ and $\ell(D_{\negalpha}(\mathbf{Q}))=\ell(D_{\negalpha-\textbf{1}}(\mathbf{Q}))+n-1$;
%\item for $i,j\in I$ with $j\neq i$, $D_{\negalpha-\textbf{1}}+Q_i+Q_j$ is a discrepancy w.r.t. $Q_i$ and $Q_j$;
\item there exists $i\in I$ such that $D_{\negalpha-\textbf{1}}(\mathbf{Q})+Q_i+Q_j$ is a discrepancy w.r.t. $Q_i$ and $Q_j$, for any $j\neq i$;
%\item there exists $i\in I$ such that $\nabla_i(\negalpha)=\emptyset$ and $\nabla_{\{i,j\}}(\negalpha)\neq \emptyset$ for every $j\neq i$.
\end{enumerate}
\end{proposition}

For $\negbeta^1,\dots, \negbeta^s\in \Z^n$, the least upper bound of $\negbeta^1,\dots, \negbeta^s$ is defined as
$$
\lub(\negbeta^1,\dots , \negbeta^s):=(\max\{\be_1^1, \dots, \be_1^s\},\dots , \max\{\be_n^1, \dots, \be_n^s\})\in \Z^n.
$$
%It is known that the set $ \widehat{\Gamma}(\mathbf{Q})$ determines $ \widehat{H}(\mathbf{Q})$ in terms of least upper bounds, see \cite[Theorem 3.4]{MTT2019}. In this way, $ \widehat{\Gamma}(\mathbf{Q})$ can be seen as a \textit{generating set} of $ \widehat{H}(\mathbf{Q})$ in the sense of \cite{G2004}. We observe that, unlike the case of generating set for Weierstrass semigroups given in \cite{G2004}, the set $ \widehat{\Gamma}(\mathbf{Q})$ is not finite. Nevertheless, it is finitely determined as follows.\\

%\begin{theorem}\cite[Theorem 3.4]{MTT2019}
%Assume $q\geq n$. The generalized Weierstrass semigroup of $F$ at $\Q$ can be written as
%$$
%\widehat{H}(\Q)=\{\lub(\negbeta^1, \negbeta^2, \dots, \negbeta^n): \negbeta^l\in \widehat{\Gamma}(\Q)\text{ for }l\in I\}.
%$$
%\end{theorem}

Let $\pi_i$ be the smallest positive integer such that $\pi_i Q_i- \pi_iQ_{i-1}$ is a principal divisor on $F$ for $i=2, \dots, n$. This positive integer will be called the \textit{period} associated to $Q_{i-1}$ and $Q_i$. For $2\leq i \leq n$, let $\negeta^i\in \Z^n$ be the $n$-tuple whose $j$-th coordinate is given by
$$
\eta^i_j=\left\{\begin{array}{ll}
-\pi_i, & \text{if }j=i-1,\\
\pi_i, & \text{if }j=i,\\
0, & \text{otherwise}.
\end{array}\right.
$$
Moreover, define the sets
\begin{equation} \label{C}
\cC(\Q):=\{(\al_1, \al_2, \dots, \al_n)\in \Z^n : 0\leq \al_i < \pi_i \text{ for }i=2, \dots, n\}
\end{equation}
and 
\begin{equation} \label{Theta}
\Theta(\Q):=\Z\negeta^2+\Z\negeta^3+\cdots+\Z\negeta^n.
\end{equation}
The following theorem characterizes the sets of absolute and relative maximal elements of the generalized Weierstrass semigroup $\widehat{H}(\Q)$.

\begin{theorem}\cite[Theorem 3.7]{MTT2019}\label{theo_relativemaximal}
Assume that $q\geq n$. Then
$$
\widehat{\Gamma}(\Q)=(\cC(\Q)\cap \widehat{\Gamma}(\Q))+\Theta(\Q)
$$
and
$$
\widehat{\Lambda}(\Q)=(\cC(\Q)\cap \widehat{\Lambda}(\Q))+\Theta(\Q).
$$
\end{theorem}

%\begin{lemma}\cite[Theorem 13.1]{LW2001}\label{lemma1}
%Let $k$ and $n$ be non-negative integers such that $n\geq 1$. Then the number of solutions of the equation
%$$
%k_1+\dots+k_n=k
%$$
%in non-negative integers is $\binom{k+n-1}{n-1}$.
%\end{lemma}

\subsection{AG codes} In \cite{S2009}, Goppa's construction of linear codes over a function field $F/\fq$ of genus $g$ is described as follows. Let $P_1,\ldots, P_N$ be pairwise distinct rational places in $F$ and $D:=P_1+\cdots +P_N$.  Consider  other  divisor  $G$ of $F$ such that $\supp(D)\;\cap\; \supp(G)=\emptyset$. Associated to the divisors $D$ and $G$ we have the differential algebraic geometry code (AG code) $C_{\Omega}(D,G)$ defined as 
$$
C_{\Omega}(D,G)=\{(\res_{P_1}(\omega),\dots, \res_{P_N}(\omega)): \omega\in \Omega(G-D)\}\subseteq \mathbb{F}_q^N,
$$
where $\res_{P_i}(\omega)$ is the residue of the Weil differential $\omega$ at $P_i$ for $i=1,\ldots,N$.
The positive integer $N$ is the length of the code $C_{\Omega}(D,G)$, its dimension as $\mathbb{F}_q$-vector space is denoted by $k_{\Omega}$, and its minimal distance by $d_{\Omega}$. The triple $[N, k_\Omega, d_\Omega]$ represent the parameters of $C_{\Omega}(D,G)$. These parameters are related by the so-called Singleton bound: $k_\Omega+d_\Omega \leq N+1$. 

\begin{proposition}\cite[Theorem 2.2.7]{S2009}\label{goppabound}
Given the AG code $C_{\Omega}(D, G)$ with parameters $[N, k_\Omega, d_\Omega]$, we have that if $2g-2<\deg(G)<N$, then
$$k_\Omega=N-\deg(G)-1+g\quad \text{and}\quad d_\Omega\geq \deg(G)-(2g-2).
$$ 
\end{proposition}

The \textit{information rate} $ R = \dfrac{k_\Omega}{N}$ and the \textit{relative minimum distance} $\delta = \dfrac{d_\Omega}{N}$ allow us to compare codes with different parameters. These parameters satisfy the inequality 
$$
R + \delta \geq 1 - \dfrac{g-1}{N}.
$$ 
It is known from coding theory that the closer $R + \delta$ is to $1$, the greater the code's capacity to transmit words and correct errors.

Now, we present a result that can be used to improve the lower bound for the minimum distance of $C_{\Omega}(D, G)$.

\begin{theorem}\cite[Theorem 3.4]{CT2005}\label{TorresCarvalho}
Let $P_1,\dots, P_N, Q_1,\dots, Q_n$ be pairwise distinct $\mathbb{F}_q$-rational places on the function field $F/\fq$ of genus $g$. Let $(\alpha_1,\dots, \alpha_n),(\beta_1,\dots, \beta_n)$ in $\mathbb{N}^n$ be such that $\alpha_i\leq \beta_i$ for $i=1,\dots, n$. Suppose each $n$-tuple $(\gamma_1,\dots, \gamma_n)$ with $\alpha_i\leq \gamma_i\leq \beta_i$ for $i=1,\dots, n$ is a pure gap at $Q_1,\dots, Q_n$. Consider the divisors  $D=P_1+\cdots+P_N$ and $G=\sum_{i=1}^n(\alpha_i+\beta_i-1)Q_i$. Then the minimum distance $d_\Omega$ of the code $C_\Omega(D,G)$ satisfy
$$ 
d_\Omega\geq \deg(G)-(2g-2)+n+\sum_{i=1}^{n}(\be_i-\al_i).
$$
\end{theorem}

\section{On maximals elements with nonnegative coordinates}

In this section we study the set of absolute and relative maximal elements on the Weierstrass semigroup $H(\Q)$ when the periods $\pi_i$ associated to the places $Q_{i-1}$ and $Q_i$ are all the same, for each $i=2,\ldots,n$. That is, there is a positive integer $\pi$ such that $\pi=\pi_{2}=\pi_3=\dots=\pi_n$. 

Set $\Gamma(\Q):=\widehat{\Gamma}(\Q)\cap \N_0^n$ and $\Lambda(\Q):=\widehat{\Lambda}(\Q)\cap \N_0^n$. All the results in this section are true for both sets $\Gamma(\Q)$ and $\Lambda(\Q)$. Thus, through this section we will use the notation $\Upsilon$, which can be replaced by $\Gamma$ or $\Lambda$.

For $k_1, \dots, k_n$ nonnegative integers define 
$$
\Upsilon_{k_1, \dots, k_n}:=\Upsilon(\Q)\cap \prod_{i=1}^{n}[k_i\pi, (k_i+1)\pi)
$$
and
$$
\w_{k_2, \dots, k_n}:=\left(-\pi\sum_{i=2}^{n}k_i, k_2\pi, \dots, k_n\pi\right).
$$
Therefore,
\begin{equation}\label{Decomposition_Upsilon}
\Upsilon(\Q)=\bigcup_{0\leq k_1, \dots, k_n}\Upsilon_{k_1, \dots, k_n}.
\end{equation}

\begin{remark} \label{empty}
Note that, if $(k_1, \ldots, k_n) \neq (j_1,\ldots , j_n)$, then $\Upsilon_{k_1, \ldots, k_n} \cap \Upsilon_{j_1, \ldots, j_n} = \emptyset$. So, the above union is disjoint. Furthermore, is well-known from Weierstrass semigroup's theory (see e.g. \cite[Lemma 2.2]{CT2005} and Riemann-Roch's Theorem in \cite{S2009}) that $a_1+ \cdots + a_n\geq 2g$ implies $(a_1, \ldots, a_n) \in H(\Q)$. So, we have that if $k_1+\cdots + k_n\geq \ceil*{\frac{2g-1}{\pi}}$, then $\Upsilon_{k_1, \ldots, k_n}=\emptyset$.
\end{remark}

In the follow we show that $\Upsilon(\Q)$ is completely determined by the sets of the form $\Upsilon_{k, 0, \dots, 0}$ with $k\geq 0$.
\begin{proposition}\label{Gammasymmetry}
Let $k_1, \dots, k_n\geq 0$. Then
$$
\Upsilon_{k_1, \dots, k_n}=\Upsilon_{k_1+\dots +k_n, 0, \dots, 0}+\w_{k_2, \dots, k_n}.
$$
In particular,
$$
|\Upsilon(\Q)|=\sum_{0\leq k}^{\ceil*{\frac{2g-1}{\pi}} -1}\binom{k+n-1}{ n-1}|\Upsilon_{k, 0, \dots, 0}|.
$$
\end{proposition}
\begin{proof}
Let $\x=((k_1+k_2+\cdots+k_n)\pi+b_1, b_2, \dots, b_n)\in \Upsilon_{k_1+ k_2+ \dots+ k_n, 0, \dots, 0}$, where $0\leq b_i<\pi$ for each $i\in \{ 1,\ldots,n \}$. Since
$$
\w_{k_2, \dots, k_n}=(k_2+\cdots+k_n)\negeta^2+(k_3+\cdots+k_n)\negeta^3+\cdots+(k_{n-1}+k_n)\negeta^{n-1}+k_n\negeta^n\in \Theta(\Q),
$$
we have, from Theorem \ref{theo_relativemaximal},
$$
\x+\w_{k_2, \dots, k_n}=(k_1\pi+b_1, k_2\pi+b_2, \dots, k_n\pi+b_n)\in \Upsilon_{k_1, \dots, k_n}.
$$
Conversely, let $\x=(k_1\pi+b_1, k_2\pi+b_2, \dots, k_n\pi+b_n)\in \Upsilon_{k_1, \dots, k_n}$, where $0\leq b_i<\pi$ for each $i\in \{1,\ldots,n\}$. From Theorem \ref{theo_relativemaximal}, there exist $\y=(\al_1, \al_2, \dots, \al_n)\in\widehat{\Upsilon}(\Q)$ with $\al_1\in \Z$ and $0\leq \al_i< \pi$ for each $i\in \{2,\ldots,n\}$, and $\z=t_2\negeta^2+t_3\negeta^3+\cdots+t_n\negeta^n$ with $t_i\in \Z$ for each $i\in \{2,\ldots,n\}$, such that $\x=\y+\z$. This implies that
$$
k_1\pi+b_1=\al_1-t_2\pi, \quad k_n\pi+b_n=\al_n+t_n\pi,\quad\text{and } 
$$ 
$$
k_i\pi+b_i=\al_i+(t_i-t_{i+1})\pi \text{ for }i=2, \dots, n-1.
$$
Hence, we obtain $t_i=k_i+k_{i+1}+\dots+k_n$ for $i=2, \dots, n$, $\al_i=\b_i$ for $i=2, \dots, n$, and $\al_1=(k_1+k_2+\cdots+k_n)\pi+b_1$, that is, 
$$
\y=((k_1+k_2+\cdots+k_n)\pi+b_1, b_2, \dots, b_n)\in \Upsilon_{k_1+\cdots+k_n, 0, \dots, 0}\quad \text{and}\quad \z=\w_{k_2, \dots, k_n}.
$$ 

To finish the proof, from (\ref{Decomposition_Upsilon}) we have 
$$
|\Upsilon(\Q)|=\sum_{0\leq k_1, \dots, k_n}|\Upsilon_{k_1, \dots, k_n}|=\sum_{0\leq k_1, \dots, k_n}|\Upsilon_{k_1+\dots +k_n, 0, \dots, 0}|=\sum_{\substack{0\leq k, k_1, \dots, k_n\\ k_1+\dots +k_n=k}}|\Upsilon_{k, 0, \dots, 0}|.
$$
From Remark \ref{empty} and using that  the number of solutions of the equation $k_1+\dots+k_n=k$
in nonnegative integers is $\binom{k+n-1}{n-1}$ (see \cite[Theorem 13.1]{LW2001}), we obtain
$$
|\Upsilon(\Q)|=\sum_{0\leq k}^{\ceil*{\frac{2g-1}{\pi}} -1}\binom{k+n-1}{n-1}|\Upsilon_{k, 0, \dots, 0}|.
$$
\end{proof}

Let $S_{n}$ be the set of permutations of the set  $\{1,\ldots, n\}$. For $\x=(x_1, x_2, \dots, x_n)\in \N_0^n$ and $\si\in S_n$, define $\si(\x):=(x_{\si(1)}, x_{\si(2)}, \dots, x_{\si(n)})$. As a consequence of the previous proposition we show that, under certain conditions, the set $\Upsilon(\Q)$ is invariant under permutations, that is, $\si(\Upsilon(\Q))=\Upsilon(\Q)$ for all $\si\in S_n$.

\begin{corollary}\label{permutationLambda}
Suppose that $x_1\equiv x_2 \equiv \cdots \equiv x_n \pmod{\pi}$ for each $\x=(x_1, x_2, \dots, x_n)\in \Upsilon(\Q)$. Then
$$
\negbeta\in \Upsilon(\Q)\quad \Rightarrow \quad \si(\negbeta)\in \Upsilon(\Q)\text{ for each }\si\in S_n.
$$
\end{corollary}
\begin{proof}
Suppose that $\negbeta\in \Upsilon_{k_1, \dots, k_n}$ for some $k_1, \dots, k_n \in \N_0$ and let $0\leq i < \pi$ be such that
$$
\negbeta=(k_1\pi+i, k_2\pi+i, \dots, k_n\pi+i).
$$
For each $\si\in S_{n}$ we have that
\begin{align*}
\si(\negbeta)&=(k_{\si(1)}\pi+i, k_{\si(2)}\pi+i, \dots, k_{\si(n)}\pi+i)\\
&=((k_{\si(1)}+k_{\si(2)}+\cdots+k_{\si(n)})\pi+i, i, \dots, i)+(-\pi\textstyle\sum_{i=2}^{n}k_{\si(i)}, k_{\si(2)}\pi, \dots, k_{\si(n)}\pi)\\
&=((k_1+k_2+\cdots+k_n)\pi+i, i, \dots, i)+\w_{\si(2), \si(3), \dots, \si(n)}.
\end{align*}
From Proposition \ref{Gammasymmetry} we have
$$
\negalpha:=((k_1+k_2+\cdots+k_n)\pi+i, i, \dots, i)\in \Upsilon_{k_1+\dots+k_n, 0, \dots, 0}=\Upsilon_{k_{\si(1)}+\dots+k_{\si(n)}, 0, \dots, 0},
$$
since $\negbeta \in \Upsilon_{k_1, \dots, k_n}$. Therefore, again from Proposition \ref{Gammasymmetry}, we conclude that 
$$
\si(\negbeta)=\negalpha+\w_{\si(2), \si(3), \dots, \si(n)}\in \Upsilon_{k_{\si(1)}, \dots, k_{\si(n)}}.
$$
\end{proof}

\section{On the Pure Gap Set}

In this section we study the pure gap set $G_0(\Q)$ of $F$ at $\Q$. Using the results from the previous section, we will present a complete description of the pure gap set in terms of the elements in the set $\Lambda^*(\Q):=\widehat{\Lambda}(\Q)\cap \N^n$ and an explicit expression for its cardinality.

For this, we start by defining the greatest lower bound of a set of elements in $\Z^n$. Given $\negbeta^1,\dots, \negbeta^s\in \Z^n$, we define their greatest lower bound by
$$
\glb(\negbeta^1,\dots , \negbeta^s):=(\min\{\be_1^1, \dots, \be_1^s\},\dots , \min\{\be_n^1, \dots, \be_n^s\})\in \Z^n.
$$

In the next result, we have a description of the pure gap set in terms of relative maximal elements.

\begin{theorem}\cite[Theorem 3.2]{TT2019}\label{puregapscharacterization}
Assume $q\geq n$. Then
$$
G_0(\Q)=\bigcup_{(\negbeta^1, \dots, \negbeta^n)\in \Lambda(\Q)^n}^{}\left(\bigcap_{i=1}^{n}\overline{\nabla}_i(\negbeta^i)\right).
$$
\end{theorem}

As a direct consequence, we obtain the following description of the pure gap set.

\begin{corollary}\label{coro_G0(Q)}
Assume $q\geq n$. Then
\begin{equation*}
G_0(\Q)=\left\{\glb(\negbeta^1, \dots, \negbeta^n): \negbeta^l\in \Lambda(\Q)\text{ and } \beta_l^l<\min\{\beta_l^1, \dots, \widehat{\beta_l^l}, \dots, \beta_l^n\}\text{ for }l\in I\right\}.
\end{equation*}
\end{corollary}
\begin{proof}
Let $\negbeta \in \N^n$. From Theorem \ref{puregapscharacterization}, we have that $\negbeta \in G_0(\Q)$ if and only if there exist $\negbeta^1, \dots, \negbeta^n\in \Lambda(\Q)$ such that $\negbeta\in \bigcap_{i=1}^{n}\overline{\nabla}_i(\negbeta^i)$. However, from \cite[Remark 3.3]{TT2019} we have that
\begin{equation*}\label{Nabla}
\bigcap_{i=1}^{n}\overline{\nabla}_i(\negbeta^{i})=\left\{
\begin{array}{ll}
\{\glb(\negbeta^1, \dots, \negbeta^n)\}, & \text{if }\beta_l^l<\min\{\beta_l^1, \dots, \widehat{\beta_l^l},\dots, \be_l^n\}\text{ for }l\in I,\\
\emptyset, & \text{otherwise}.
\end{array} \right. 
\end{equation*}
This implies that $\negbeta\in G_0(\Q)$ if and only if there exist $\negbeta^1, \dots, \negbeta^n\in \Lambda(\Q)$ such that $\beta_l^l<\min\{\beta_l^1, \dots, \widehat{\beta_l^l},\dots, \be_l^n\}$ for $l\in I$ and $\negbeta=\glb(\negbeta^1, \dots, \negbeta^n)$.
\end{proof}

\begin{remark}
Let $\negbeta\in G_0(\Q)\subseteq \N^n$. Then there exist $\negbeta^1, \dots, \negbeta^n\in \Lambda(\Q)$ such that $\beta_l^l<\min\{\beta_l^1, \dots, \widehat{\beta_l^l}, \dots, \beta_l^n\}$ for $l\in I = \{1,\ldots,n\}$ and $\negbeta=\glb(\negbeta^1, \dots, \negbeta^n)=(\be_1^1, \dots, \be_n^n)\in \N^n$. Now, we claim that $\negbeta^l\in \Lambda^*(\Q)$ for each $l\in I$. Indeed, suppose for a moment that $\be_j^i=0$ for some $i, j\in I$. If $i=j$ then $\negbeta\notin \N^n$, a contradiction. If $i\neq j$ then
$$
\be_j^j<\min\{\be_j^1, \dots, \widehat{\be_j^j}, \dots, \be_j^n\}=0,
$$
a contradiction. Therefore $\be_j^i\geq 1$ for all $i, j\in I$. This implies that we can rewrite the pure gap set $G_0(\Q)$ as
%\begin{equation}\label{puregapset}
$$
G_0(\Q)=\left\{\glb(\negbeta^1, \dots, \negbeta^n): \negbeta^l\in \Lambda^*(\Q)\text{ and } \beta_l^l<\min\{\beta_l^1, \dots, \widehat{\beta_l^l}, \dots, \beta_l^n\}\text{ for }l\in I\right\}.
$$
%\end{equation}
%where $\Lambda^*(\Q)=\widehat{\Lambda}(\Q)\cap \N^n$.
\end{remark}

For $k_1, \dots, k_n$ nonnegative integers define the set  
$$
G_{k_1, \dots, k_n}:=G_0(\Q)\cap \prod_{i=1}^{n}[k_i\pi, (k_i+1)\pi).
$$
Therefore,
\begin{equation}\label{PureGapsDecomposition}
G_0(\Q)=\bigcup_{0\leq k_1, \dots, k_n}G_{k_1, k_2, \dots, k_n}.
\end{equation}

From Corollary \ref{coro_G0(Q)} we have $\negbeta\in G_{k_1, \dots, k_n}$ if and only if there exist elements $\negbeta^1, \dots, \negbeta^n\in \Lambda(\Q)$ such that 
\begin{itemize}
\item [$\bullet$] $\beta_l^l<\min\{\beta_l^1, \dots, \widehat{\beta_l^l}, \dots, \beta_l^n\}$ for $l\in I$,
\item [$\bullet$] $\negbeta= \glb(\negbeta^1, \dots, \negbeta^n)=(\be_1^1, \be_2^2, \dots, \be_n^n)$, and
\item [$\bullet$] $\be_l^l\in [k_l\pi, (k_l+1)\pi)$ for $l\in I$.
\end{itemize}
Suppose that $\negbeta^l\in \Lambda_{s_1^l, s_2^l, \dots, s_n^l}$ for each $l\in I$. Then $\be_j^l\in [s_j^l\pi, (s_j^l+1)\pi)$ for $j\in I$. Since $\be_l^l\in [k_l\pi, (k_l+1)\pi)$, we obtain $k_l=s_l^l$. Furthermore, since $\be_j^j< \be_j^l$ for $j\in I\setminus \{l\}$, we deduce that $s_j^l\geq s_j^j=k_j$. Thus we conclude that $\negbeta^l\in \Lambda_{s_1^l, \dots, s_n^l}$, where $k_j\leq s_j^l$ for $j\in I$ and $k_l=s_l^l$. Therefore, we can rewrite the set $G_{k_1,\dots, k_n}$ as
\begin{equation}\label{Gk1...km}
G_{k_1,\dots, k_n}=\left\{\glb(\negbeta^1, \dots, \negbeta^n): 
\begin{array}{l}
\negbeta^l\in \Lambda_{s^l_1, \dots, s^l_n}, \text{ where }k_j \leq  s_j^l \text{ for }j\in I \text{ and } k_l=s_l^l, \text{ and} \\
\beta_l^l<\min\{\beta_l^1, \dots, \widehat{\beta_l^l}, \dots, \beta_l^n\}\text{ for }l\in I
\end{array}\right\}.
\end{equation}

The following result shows that to determine $G_0(\Q)$ it is enough to determine the sets of the form  $G_{k, 0, \dots, 0}$ with $k\geq 0$.

\begin{proposition}\label{Puregapsymmetry}
Let $k_1,\dots, k_n \geq 0$. Then
$$
G_{k_1,\dots, k_n}=G_{k_1+\cdots+k_n, 0, \dots, 0}+\w_{k_2, \dots, k_n}. 
$$
In particular, 
\begin{equation}\label{CardinalityPuregaps}
|G_0(\Q)|= \sum_{0\leq k}\binom{k+n-1}{n-1}|G_{k, 0, \dots, 0}|.
\end{equation}
\end{proposition}
\begin{proof}
Let $\negbeta\in G_{k_1+\cdots+k_n, 0, \dots, 0}$. From the description given in (\ref{Gk1...km}), there exist elements $\negbeta^1, \dots, \negbeta^n\in \Lambda(\Q)$ satisfying the following conditions:
\begin{itemize}
\item [$\bullet$] $\negbeta^1\in \Lambda_{k_1+\cdots+k_n, s_2^1, \dots, s_n^1}$, where $s_i^1\geq 0$ for $i\in I\setminus \{1\}$;
\item [$\bullet$] $\negbeta^l\in \Lambda_{s_1^l, \dots, s_n^l}$ for $l\in I\setminus \{1\}$, where $s_1^l\geq k_1+\cdots+k_n$, $s_l^l=0$, and $s_i^l\geq 0$ for $i\in I\setminus\{1, l\}$;
\item [$\bullet$] $\beta_l^l<\min\{\beta_l^1, \dots, \widehat{\beta_l^l}, \dots, \beta_l^n\}$ for $l\in I$; and
\item [$\bullet$] $\negbeta=\glb(\negbeta^1,\dots, \negbeta^n)$.
\end{itemize}
These conditions and Proposition \ref{Gammasymmetry} guarantee that  
$$
\negalpha^1:=\negbeta^1+\w_{k_2, \dots, k_n}\in \Lambda_{k_1, s_2^1+k_2, \dots, s_n^1+k_n}
$$
and, for $l\in I\setminus \{1\}$,
$$
\negalpha^l:=\negbeta^l+\w_{k_2, \dots, k_n}\in \Lambda_{t_1^l, t_2^l, \dots, t_n^l},
$$
where 
$$
t_i^l=\left\{
\begin{array}{ll}
s_1^l, & \text{if }i=1,\\
k_l, & \text{if }i=l,\\
s_i^l+k_i, & \text{if } i\in I\setminus\{1, l\}.
\end{array}\right.
$$
This implies that $\glb(\negalpha^1, \dots, \negalpha^n)\in \prod_{i=1}^{n}[k_i\pi, (k_i+1)\pi)$. Furthermore, since
$$
\al_l^l<\min\{\al_l^1, \dots, \widehat{\al_l^l}, \dots, \al_l^n\}\quad \Leftrightarrow\quad \be_l^l<\min\{\be_l^1, \dots, \widehat{\be_l^l}, \dots, \be_l^n\},
$$
we obtain that $\glb(\negalpha^1, \dots, \negalpha^n)\in G_0(\Q)$ and therefore $\glb(\negalpha^1, \dots, \negalpha^n)\in G_{k_1, \dots, k_n}$. Now, note that
\begin{align*}
\negbeta+\w_{k_2, \dots, k_n}&=\glb(\negbeta^1,\dots, \negbeta^n)+\w_{k_2, \dots, k_n}\\
&=\glb(\negbeta^1+\w_{k_2, \dots, k_n},\dots, \negbeta^n+\w_{k_1, \dots, k_n})\\
&=\glb(\negalpha^1,\dots, \negalpha^n).
\end{align*}
It follows that $G_{k_1+\cdots+k_n, 0, \dots, 0}+\w_{k_2, \dots, k_n}\subseteq G_{k_1, \dots, k_n}$. The other inclusion is proved analogously. 

To finish the proof, from (\ref{PureGapsDecomposition}) we have
$$
|G_0(\Q)|=\sum_{0\leq k_1, \dots, k_n}|G_{k_1, \dots, k_n}|=\sum_{0\leq k_1, \dots, k_n}|G_{k_1+\dots +k_n, 0, \dots, 0}|=\sum_{\substack{0\leq k, k_1, \dots, k_n\\ k_1+\dots +k_n=k}}|G_{k, 0, \dots, 0}|.
$$
So, we can conclude that
$$
|G_0(\Q)|=\sum_{0\leq k}\binom{k+n-1}{n-1}|G_{k, 0, \dots, 0}|.
$$
\end{proof}

In the next section we will present an application over Kummer extensions of the results presented in this section. A particular case can be seen in Example \ref{ex numerico}.

Next, we show that the properties of set $\Lambda(\Q)$ described in Corollary \ref{permutationLambda} carry over to pure gap set $G_0(\Q)$.

\begin{proposition}\label{Gapspermutation}
Suppose that $x_1\equiv x_2 \equiv \cdots \equiv x_n \pmod{\pi}$ for each $\x=(x_1, x_2, \dots, x_n)\in \Lambda(\Q)$. Then
$$
\negbeta\in G_0(\Q)\quad \Rightarrow \quad\si(\negbeta)\in G_0(\Q)\text{ for each }\si\in S_n.
$$
\end{proposition}
\begin{proof}
Let $\negbeta\in G_0(\Q)$. Then there exist elements $\negbeta^1, \dots, \negbeta^n\in \Lambda(\Q)$ such that $\beta_l^l<\min\{\beta_l^1, \dots, \widehat{\beta_l^l}, \dots, \beta_l^n\}$ for $l\in I = \{1,\ldots,n\}$ and $\negbeta=\glb(\negbeta^1, \dots, \negbeta^n)=(\be_1^1, \be_2^2, \dots, \beta_n^n)$. From Corollary \ref{permutationLambda}, for each $\si\in S_n$ we have that $\si(\negbeta^l)\in \Lambda(\Q)$. Furthermore,
$$
\beta_l^l<\min\{\beta_l^1, \dots, \widehat{\beta_l^l}, \dots, \beta_l^n\} \text{ for }l\in I\, \Leftrightarrow \, \beta_{\si(l)}^{\si(l)}<\min\{\beta_{\si(l)}^{1}, \dots, \widehat{\beta_{\si(l)}^{\si(l)}}, \dots, \beta_{\si(l)}^{n}\}\text{ for }l\in I.
$$
This implies that $\si(\negbeta)=(\be_{\si(1)}^{\si(1)}, \be_{\si(2)}^{\si(2)}, \dots, \beta_{\si(n)}^{\si(n)})=\glb(\si(\negbeta^1), \si(\negbeta^2), \dots, \si(\negbeta^n))\in G_0(\Q)$.
\end{proof}

\begin{corollary}\label{corollary-permutation-Gk0..0}
Suppose that $x_1\equiv x_2 \equiv \cdots \equiv x_n \pmod{\pi}$ for each $\x=(x_1, x_2, \dots, x_n)\in \Lambda(\Q)$. Then
$$
(k\pi+i_1, i_2, \dots, i_n)\in G_{k, 0, \dots, 0}\quad \Rightarrow \quad (k\pi+i_{\si(1)}, i_{\si(2)}, \dots, i_{\si(n)})\in G_{k, 0, \dots, 0}
$$
for each $\si\in S_n$.
\end{corollary}
\begin{proof}
Let $(k\pi+i_1, i_2, \dots, i_n)$ be an element in $G_{k, 0, \dots, 0}$, $\sigma$ be an element in $S_n$, and $s\in I$ be such that $\sigma(s)=1$. From Proposition \ref{Gapspermutation} we obtain 
$$
\sigma((k\pi+i_1, i_2, \dots, i_n))=(i_{\sigma(1)}, i_{\sigma(2)}, \dots, i_{\sigma(s-1)}, km+i_1, i_{\sigma(s+1)}, \dots, i_{\sigma(n)})\in G_{0, 0, \dots, k, \dots, 0}.
$$
Now, define $\w=(-k\pi, 0, \dots, \underbrace{k\pi}_{s\text{-th entry}}, \dots, 0)$. From Proposition \ref{Puregapsymmetry} we deduce that
$$
(km+i_{\sigma(1)}, i_{\sigma(2)}, \dots, i_{\sigma(n)})=\sigma((k\pi+i_1, i_2, \dots, i_n))-\w\in G_{k, 0, \dots, 0}.
$$
\end{proof}

\section{Set of pure gaps at several rational places on Kummer Extensions}

In this section we apply the results in the previous section to provide an explicit description of the set of pure gaps $G_0(\Q)$ at several rational places distinct to the infinity place on Kummer extensions and its cardinality. For this we completely determine the sets of absolute and relative maximal elements of the generalized Weierstrass semigroup $\widehat{H}(\Q)$, where each $Q_i$ is distinct to the infinity place. We conclude this section with some applications in coding theory by presenting differential AG codes in many places over Kummer extensions. In particular, we present some AG codes which have better relative parameters than the AG codes obtained in \cite[Examples 1 and 3]{HY2018} and improve the parameters given in MinT’s Tables \cite{MinT}.

\subsection{Kummer extensions} Consider the Kummer extension defined by the affine equation
\begin{equation}\label{KummerEquation}
\cX:\quad y^m=\prod_{j=1}^{r}(x-\al_j)^\la,
\end{equation}
where $\al_1, \dots, \al_r$ are pairwise distinct elements in $K$, where $K$ is the algebraic closure of a finite field $\fq$, $\lambda \in \N$, and $m, r\geq 2$ are integers such that $\char(K)\nmid m$ and $\gcd(m, \lambda r)=1$. Let $F=K(\mathcal{X})$ be its function field. Then the genus of $F$ is given by $g=\frac{1}{2}(m-1)(r-1)$ and $W=(2g-2)P_\infty=(rm-r-m-1)P_\infty$ is a canonical divisor of $F$, where $P_\infty$ denote the single place at infinity of $F$. For $2\leq n \leq r$, let $P_1, P_2, \dots, P_n$ be the totally ramified places of the extension $F/K(x)$  corresponding to $\al_1, \al_2, \dots, \al_n$, respectively. From \cite[Proposition 4.9]{CMT2024}, the period associated to the places $P_i$ and $ P_j$ is $m$ for any $1\leq i, j\leq n$ such that $i\neq j$. In this subsection, using the results of the previous sections, we determine the cardinality and a simple explicit description of the pure gap set $G_0(\P)$ over $F$, where $\P=(P_1, \dots, P_n)$, which is different from the description given by Hu and Yang in \cite[Theorem 3]{HY2018}.

The following result describes some divisors of a Kummer extension.

\begin{proposition}\cite[Proposition 5]{YH2017}\label{functions Yang Hu}
Let $F/K(x)$ be the Kummer extension given in (\ref{KummerEquation}). Then we have the following divisors:
\begin{enumerate}[\rm (i)]
\item $(x-\alpha_j) = mP_j - mP_{\infty}$, for $1 \leq j \leq r$;
%(2) (y) = \lambda P_1 + \cdots + \lambda P_r - r \lambda P_{\infty};
\item $(z) = P_1 + \cdots + P_r - rP_{\infty}$, where $z := y^Af(x)^B$, and $A, B$ are integers such that $A\lambda + Bm = 1$.
\end{enumerate}
\end{proposition}

In the next results we have information about the Weierstrass semigroup at one and several places on Kummer extensions.

\begin{theorem}\cite[Theorem III.2]{CMQ2016} \label{HP1}
Let $F/K(x)$ be the Kummer extension given in (\ref{KummerEquation}). Suppose that $g\geq 1$ and that $P$ is a totally ramified place in the extension $F/K(x)$. Then
\[H(P)=
\begin{cases}
\mathbb{N}_0 \setminus \left\{1+i+mj \,:\, 0 \leq i \leq m-2-\floor*{\frac{m}{r}}, \, 0\leq j \leq r-2- \floor*{ \frac{r(i+1)}{m}}\right \}, &\text{ if } P\ne P_\infty, \\
\langle m,r\rangle, &\text{ if } P=P_\infty.
\end{cases}\]
\end{theorem}

\begin{theorem}\cite[Theorem 10]{YH2017}\label{absmaxKummer}
Let $F/K(x)$ be the Kummer extension given in (\ref{KummerEquation}) and assume that $g \geq 1$. Let $\Gamma^*(\P):=\widehat{\Gamma}(\P)\cap \N^n$. Then, for $2 \leq n \leq r - \floor*{\frac{r}{m}}$,
\begin{align*}
\Gamma^*(\P)=\Bigg\{&(k_1m+i, k_2m+i, \dots, k_nm+i)\in \N^n : 1\leq i \leq m-1-\floor*{\frac{m}{r}},\\
&\quad k_1, k_2, \dots, k_n\geq 0, \, k_1+\dots+k_n=r-n-\floor*{\frac{ri}{m}}\Bigg\},
\end{align*}
and for $r-\floor*{\frac{r}{m}}<n\leq r$, 
$$
\Gamma^*(\P)=\emptyset.
$$
\end{theorem}

\subsection{Maximal elements in $\widehat{H}(\P)$}

Our aim in this section is to compute explicitly the absolute and relative maximal elements in the generalized Weierstrass semigroup $\widehat{H}(\P)$. We observe that
\begin{equation} \label{Cm}
\cC (\P)=\{\negbeta\in \mathbb{Z}^{n} \ : \ 0\leq \beta_i< m \ \mbox{for } i=2,\ldots,n\},
\end{equation}
and $\Theta(\P)$ is generated by the $n$-tuples
\begin{equation}\label{eta_i}
\negeta^i=(0,\ldots,0,-m,\underbrace{m}_{i\text{-th entry}},0,\ldots,0)\in \mathbb{Z}^{n} \mbox{ for } i=2,\ldots,n.
\end{equation}

\begin{proposition} \label{discrepancia m upla}
Let $i \in \{1,\ldots, m-1\}$ and $2 \leq n \leq r - \floor*{\frac{r}{m}}$. Then the divisor $ A= (km+i)P_1 + \sum_{\ell=2}^{n}iP_{\ell}$, where $k=r-n-\floor*{\frac{ri}{m}}$, is a discrepancy to every pair of distinct places in $\{ P_1, \ldots,P_{n}\}$.
\end{proposition}
\begin{proof}
First, note that $\cL(A) \neq \cL(A - P_j)$ each for $j \in \{ 1,\ldots,n\}$, since
$$
f_1= \dfrac{\prod_{t =n+1}^{r}(x-\alpha_t)}{(x-\alpha_1)^k z^{i}} \in \cL(A) \setminus \cL(A-P_j).
$$
Next, consider the canonical divisor $W = (rm-r-m-1)P_{\infty}$. For $j \in \{ 2,\ldots,n\}$, 
$$
f_2= (x - \alpha_1)^k z^{i-1} \prod_{t=2 \atop t\neq j}^{n}(x - \alpha_t) \in \cL(W-A+P_1+P_j) \setminus \cL(W-A+P_1).$$
So, by Lemma \ref{lemma noether}, we conclude that $\cL(A-P_1) = \cL(A-P_1-P_j)$. 

Finally, for $j,s \in \{ 2,\ldots,n\}$ with $j \neq s$, we have that
$$
f_3 = (x - \alpha_1)^{k+1} z^{i-1} \prod_{t=2 \atop t\neq j,s}^{n}(x - \alpha_t) \in \cL(W-A+P_s+P_j) \setminus \cL(W-A+P_s).$$
So, again by Lemma \ref{lemma noether}, we conclude that $\cL(A-P_s) = \cL(A-P_s-P_j)$. 

\end{proof}

\begin{lemma} \label{nabla vazio}
Let $i \in \{1, \ldots , m-1\}$ and $k = r-n-\floor*{\frac{ri}{m}}$. Then
$$
\nabla_{1}^{n}((km+i,m+ i - 1,\ldots,m+i-1,i-1))= \emptyset.
$$
\end{lemma}
\begin{proof}
Let $\negbeta = (km+i, m+ i - 1, \dots, m+i-1, i-1)$. Note that, 
$$
\displaystyle f= (x-\alpha_1)^k z^{i-1} \prod_{j=2}^{n-1} (x - \alpha_j) \in \cL(W - D_{\negbeta}(\P) + P_1) \setminus \cL(W-D_{\negbeta}(\P)).
$$
So, by Lemma \ref{lemma noether}, we have $\ell(D_{\negbeta}(\P)) = \ell(D_{\negbeta}(\P) - P_1)$, and the result follows from Proposition \ref{prop 2.1.3 da tese}.

\end{proof}

In the next result we will compute $\widehat{\Gamma}(\mathbf{P})\cap \mathcal{C}(\P)$, which, according to Theorem \ref{theo_relativemaximal}, determines entirely $\widehat{\Gamma}(\mathbf{P})$.

\begin{theorem} \label{gamma hat} Let $2 \leq n \leq r - \floor*{\frac{r}{m}}$ and let $P_1,\ldots,P_{n}$ be as above. Let
$$\widehat{T}_n=\left\{(km+i,i,\ldots,i)\in \mathbb{Z}^n \ : \ i=1,\ldots,m-1, \mbox{ with } k=r-n-\floor*{\frac{ri}{m}} \right\}\cup \{{\bf 0}\}.$$
Then  $$\widehat{\Gamma}(\P)\cap \mathcal{C}(\P)=\widehat{T}_n.$$
\end{theorem}
\begin{proof} 
From Propositions \ref{discrepancia m upla} and \ref{prop equiv discrepancia} it follows that $\widehat{T}_n \subseteq \widehat{\Gamma}(\P)\cap \mathcal{C}(\P)$. 

To prove the reverse inclusion we will use induction on $n$. 
Define
$${\negalpha}^{i,n}:=(km+i,i,\ldots,i)\in \mathbb{Z}^n\;, \mbox{ where } k=r-n-\floor*{\frac{ri}{m}}.$$
Thus $\widehat{T}_n=\{\negalpha^{i,n}:i=1,\ldots,m-1\}\cup\{\bf 0\}$. 

For $n=2$, let $(\mu,i) \in \widehat{\Gamma}(P_1, P_2)\cap \mathcal{C}(P_1, P_2)$. Thus, $0 < i < m$ and from Propositions \ref{discrepancia m upla} and \ref{prop equiv discrepancia}, we have $\negbeta =((r-2-\floor*{\frac{ri}{m}})m+i,i) \in \widehat{\Gamma}(P_1, P_2)\cap \cC(P_1, P_2)$. Then, by Proposition \ref{absmax}, $\nabla^2_2((\mu,i)) = \{(\mu,i)\}$ and $\nabla^2_2(\negbeta) = \{\negbeta\}$. Suppose that $(\mu,i) \neq \negbeta$. If $\mu < (r-2-\floor*{\frac{ri}{m}})m+i$, then  $(\mu,i) \in \nabla^2_2(\negbeta)$, a contradiction. And, if $\mu > (r-2-\floor*{\frac{ri}{m}})m+i$, then $\negbeta \in \nabla^2_2((\mu,i))$, a contradiction. So, $(\mu,i) = ((r-2-\floor*{\frac{ri}{m}})m+i,i)$, and we have that the result is true for $n=2$.

Let now $n\geq 3$ and suppose that $\widehat{\Gamma}(P_1,\ldots,P_{k})\cap \cC(P_1, \dots, P_k)=\widehat{T}_{k}$ for $k=2,\ldots,n-1$. Note that $\alpha_1^{i,n} = \alpha_1^{i,n-j} - jm$ for any $j \in \{1,\ldots, n-1\}$.

Let $\mathbf{0} \neq \neglambda = (\lambda_1, \ldots , \lambda_n) \in \widehat{\Gamma}(\P)\cap \cC(\P)$.  If $\lambda_1 > 0$, then the result follows from Theorem \ref{absmaxKummer}. If $\lambda_1=0$, then $\tilde{\neglambda}= (\lambda_2, \ldots , \lambda_n) \in H(P_2, \ldots , P_n)$ with $\lambda_j < m$ for each $j=2,\ldots,n$. By Theorem \ref{absmaxKummer}, if $\negbeta = (\beta_2,\ldots,\beta_{n}) \in \Gamma(P_2, \ldots, P_n)$ with $\beta_\ell < m$ for any $\ell \in \{2,\ldots,n\}$, then $\negbeta = (i,\ldots,i)$. So, for an absolute maximal $\negbeta^{(1)} = (i^{(1)},\ldots,i^{(1)}) \in \nabla_{1}^{n-1}(\tilde{\neglambda})$, we get $\lambda_2 = i^{(1)}$. Now, for an absolute maximal $\negbeta^{(2)} = (i^{(2)},\ldots,i^{(2)}) \in \nabla_{2}^{n-1}(\tilde{\neglambda})$, we get $\lambda_3 = i^{(2)}$. Hence $i^{(1)}=i^{(2)}$. Then, we get $\tilde{\neglambda} = (i,\ldots,i)$ and so $\neglambda = (0, i \ldots , i) \in H(\P)$. Thus there is an absolute maximal $\negbeta \in \Gamma(\P)$ such that $\negbeta \in \nabla_{1}^{n}(\neglambda)$, a contradiction by Theorem \ref{absmaxKummer}, since there is no element in $\Gamma(\P)$ with this characteristic. Therefore we can suppose $\lambda_1 < 0$.

Note that by Proposition \ref{absmax} $\ell(D_{\neglambda}(\P)) \geq 1$. Thus $\deg(D_{\neglambda}(\P)) \geq 0$ and so $\textstyle \lambda_1 \geq - \sum_{j=2}^{n} \lambda_j \geq - (n-1)(m-1)$. Then, we have that $\lambda_1 + (n-1)m \geq 0$. Let $s:=\min \{ t \in \mathbb{N}_0 : \lambda_1 + tm \geq 0 \}$. Thus $1 \leq s \leq n-1$.

$\bullet$ If $1 \leq s \leq n-2$, then $2 \leq n-s \leq n-1$. Without loss of generality assume that $\lambda_2 = \min \{\lambda_j : 2 \leq j \leq n\}$. Hence Theorem \ref{theo_relativemaximal} and (\ref{eta_i}) lead to
$$
\neglambda' = (\lambda_1 + sm, \lambda_2, \ldots, \lambda_{n-s}, \lambda_{n-s+1} -m , \ldots, \lambda_n - m) \in \widehat{H}(\P).
$$

So, $\negbeta = \mbox{lub} (\neglambda', \mathbf{0}) = (\lambda_1 + sm, \lambda_2, \ldots, \lambda_{n-s},0,\ldots,0) \in \widehat{H}(\P)$, and we have that $\tilde{\negbeta} = (\lambda_1 + sm, \lambda_2, \ldots, \lambda_{n-s})  \in \widehat{H}(P_1, \ldots , P_{n-s})$. In particular, $\nabla_{2}^{n-s}(\tilde{\negbeta}) \neq \emptyset$. From the induction hypothesis, there exists $\negalpha^{i, n-s} \in \widehat{T}_{n-s}$ such that $\negalpha^{i, n-s} \in \nabla_{2}^{n-s}(\tilde{\negbeta})$. Since $\alpha_{1}^{i,n} = \alpha_{1}^{n-s} -sm$ and $\alpha_2 = i \leq \alpha_j$, for each $j=3,\ldots,n$, we have that $\negalpha^{i,n} \in \nabla_{2}^{n} (\neglambda)$. Therefore, $\neglambda = \negalpha^{i,n}$.

\medskip

$\bullet$ If $s=n-1$, by using a similar argument, we obtain $\lambda_1 + (n-1)m \in H(P_1)$. By definition of $s$, we have that $\lambda_1 + (n-1)m < m$. By Theorem \ref{HP1}, $H(P_1) \cap \{0,1,\ldots,m\} = \{ 0,m - \lfloor \frac{m}{r}\rfloor, \ldots, m \}$. Thus $\lambda_1 + (n-1)m =0$ or $\lambda_1 + (n-1)m = m - j$, with $j \in \{ 1, \ldots , \lfloor \frac{m}{r}\rfloor \}$. Note that, if $\lambda_1 + (n-1)m =0$, then $\alpha_1 = -(n-1)m$, a contradiction, since $\alpha_1 \geq - (n-1)(m-1)$. 

Now, let $\lambda_1 + (n-1)m = m - j$, with $j \in \{ 1, \ldots , \lfloor \frac{m}{r}\rfloor \}$. Observe that in this case $m>r$ and we can write $\lambda_1 = -(n-1)m + m - (\lfloor \frac{ri}{m} \rfloor m + 2m - rm - i) = (r-n-\lfloor \frac{ri}{m} \rfloor)m + i$, for some $i \in \{ m - \lfloor \frac{m}{r} \rfloor , \ldots, m-1 \}$.  So, $\lambda_1 = \alpha_{1}^{i,n}$.

Suppose that $\neglambda$ and $\negalpha^{i,n}$ are not comparable in the partial order $\leq$. Without loss of generality, we may assume that $\lambda_n<\alpha_n^{i,n}$. In this way, since $\alpha^{i,n}_k \geq 1$ for all $k=2,\ldots,n$, we have that $\neglambda \in \nabla_{1}^{n}((\alpha^{i,n}_1,\alpha^{i,n}_2 +m - 1,\ldots,\alpha^{i,n}_{n-1}+m-1,\alpha^{i,n}_n - 1))$, a contradiction by Lemma \ref{nabla vazio}. Thus $\neglambda$ and $\negalpha^{i,n}$ are comparable with respect to $\leq$, which leads to $\neglambda=\negalpha^{i,n}$ by their absolute maximality. Therefore, $\widehat{\Gamma}(\P)\cap \cC(\P) \subseteq \widehat{T}_n$.
\end{proof}

In the next result, we will explicit the relative maximal elements of $\widehat{H}(\P)$  in the region $\mathcal{C}(\P)$. %According to Theorem \ref{theo_relativemaximal}, these elements determine all relative maximal elements of $\widehat{H}(\P)$ and using the results in the previous section we can determine the pure gap set $G_0(\P)$.

\begin{theorem}\label{relativemaxKummer}
Let $F/K(x)$ be the Kummer extension given in (\ref{KummerEquation}) with genus $g \geq 1$. For $2 \leq n \leq r - \floor*{\frac{r}{m}}$, let
$$\begin{array}{rcll}
\negbeta^{0,n} & = & \left((n-2)m,0,\ldots,0\right), & \mbox{and} \\
\negbeta^{i,n} & = & \left((k+n-2)m + i,i,\ldots,i\right), & \mbox{for } 1 \leq i\leq m-1, \mbox{ } k=r-n-\floor*{\frac{ri}{m}}.
\end{array}
$$
Then
$$\widehat{\Lambda}(\P) \cap \cC(\P)=\left\{\negbeta^{i,n} \ : \ 0 \leq i\leq m-1 \right\}.$$
\end{theorem}
\begin{proof} 
If $n=2$, then by Propositions \ref{absmax} and \ref{relmax} $\widehat{\Gamma}(\P)=\widehat{\Lambda}(\P)$ and so the result follows from Theorem \ref{gamma hat}.
Suppose $n\geq 3$. In order to verify that $\negbeta^{0,n}\in \widehat{\Lambda}(\P)$, we will prove that $\mathcal{L}(A)\neq \mathcal{L}(A-P_1)$ and $\mathcal{L}(W-A+P_1+P_j)\neq \mathcal{L}(W-A+P_j)$ for each $j\in \{2,\ldots,n\}$, where $A=D_{\negbeta^{0,n}-\textbf{1}+\textbf{e}_1+\textbf{e}_j}(\P)$. Observe that
$$
A=(n-2)m P_1-\sum_{t=2 \atop t\neq j}^{n}P_t
$$
and therefore
$$
(x - \alpha_1)^{2-n}\prod_{\substack{t=2\\t\neq j}}^{n}(x - \alpha_t) \in \mathcal{L}(A)\backslash \mathcal{L}(A-P_1).$$
Moreover, since $W-A+P_1+P_j= (rm - r - m - 1)P_{\infty} - (n-2)m P_1 + \sum_{t=1}^n P_t$, we get
$$
z^{-1}(x - \alpha_1)^{n-2} \prod_{t=n+1}^{r}(x - \alpha_t)\in \mathcal{L}(W-A+P_1+P_j)\backslash \mathcal{L}(W-A+P_j)
$$
and so we conclude that $\negbeta^{0,n} \in \widehat{\Lambda}(\P)\cap \cC(\P)$.

Let us denote $R_n=\left\{\negbeta^{i,n} \ : \ 1 \leq i \leq m-1 - \floor*{\frac{m}{r}}\right\}$. First, we will prove that $R_n\subseteq \widehat{\Lambda}(\P) \cap \cC(\P)$. By Proposition~\ref{relmax}, it is sufficient to prove that $A=D_{\negbeta^{i,n}-\mathbf{1}+\nege_1+\nege_j}(\P)$, with $1 \leq i \leq m-1 - \lfloor \frac{m}{r} \rfloor$, is discrepancy with respect to $P_1$ and $P_j$ for any $j\in I\backslash\{1\}$. 

Note that,
$$
A= [(k+n-2)m + i] P_1 + iP_j + (i-1) \sum_{\substack{t=2\\ t\neq j}}^{n} P_t.
$$
Let $\displaystyle h_j= \dfrac{(x - \alpha_j)^k \prod_{\substack{t=2\\ t\neq j}}^r (x - \alpha_t)^{k+1}}{(x-\alpha_1)^{n-2}z^{km + i}}$. Then, by Proposition \ref{functions Yang Hu}, we have that
\begin{align*}
(h_j) &= \displaystyle  k m P_j - k m P_{\infty} + (k+1)m \sum_{\substack{t=2\\ t\neq j}}^{r} P_t - (r-2)(k+1)m P_{\infty}\\
&\quad -(n-2)m P_1 + (n-2)m P_{\infty} - (km + i) \sum_{t=1}^{r} P_t + (km + i)r P_{\infty}\\
&\displaystyle = - [(k+n-2)m + i]P_1 -iP_j + (m - i) \sum_{\substack{t=2\\ t\neq j}}^{r} P_t + [(k+n-r)m + ri]P_{\infty}.
\end{align*}
Thus, $\displaystyle (h_j) + A = (m - 1)\textstyle\sum_{\substack{t=2\\ t\neq j}}^{r} P_t + [(k+n-r)m + ri]P_{\infty}$. So, $h_j \in \mathcal{L}(A)\setminus \mathcal{L}(A - P_j)$ since $k=r-n-\floor*{\frac{ri}{m}}$ and $m \geq 1$. This implies that $\cL(A)\neq\cL(A-P_j)$ for any $j \in \{2,\ldots,n\}$. Now, we must prove that $\mathcal{L}(A-P_1)=\mathcal{L}(A-P_1-P_j)$. By Lemma \ref{lemma noether}, it suffices to prove that $\mathcal{L}(W-A+P_1+P_j)\neq \mathcal{L}(W-A+P_1)$. We have
$$W-A+P_1+P_j=(r m - r - m - 1)P_{\infty} - [(k+n-2)m + i-1] P_1 - (i-1) \sum_{\substack{t=2}}^{n} P_t.$$
Let $\displaystyle h = \dfrac{(x - \alpha_1)^{n-2} z^{km + i-1}}{\prod_{t=2}^{r}(x-\alpha_t)^k}$. By Proposition \ref{functions Yang Hu}, we have that
\begin{align*}
(h) &= (n-2)m P_1 - (n-2)m P_{\infty} + (km +i - 1) \sum_{t=1}^{r}P_t - (km + i - 1)r P_{\infty}\\
 & \quad  - km  \sum_{t=2}^{r} P_t + km (r-1)P_{\infty}\\
 & \displaystyle = [(k+n-2)m + i - 1]P_1 + (i-1) \sum_{t=2}^{r} P_t - [(k+n-2)m +ri - r]P_{\infty}.
\end{align*}
Thus, 
$$ (h) + W-A+P_1+P_j   \displaystyle = (i-1) \sum_{t=n+1}^{r} P_t + \left[\left(1+ \floor*{\frac{ri}{m}}\right)m - ri - 1 \right]P_{\infty}.$$ 
So, $h\in \mathcal{L}(W-A+P_1+P_j)\backslash \mathcal{L}(W-A+P_1)$, since $(1+ \lfloor \frac{ri}{m} \rfloor)m - ri - 1 \geq 0$. Therefore, we have that $A$ is a discrepancy with respect to $P_1$ and $P_j$ for each $j\in \{2,\ldots,n\}$ and, by Proposition \ref{relmax}, we get $\negbeta^{i,n}\in \widehat{\Lambda}(\P)$ for any $1\leq i \leq m-1 - \lfloor \frac{m}{r} \rfloor$. Thus, we conclude that $R_n\subseteq \widehat{\Lambda}(\P) \cap \cC(\P)$. 

Now, let $\negbeta\in \widehat{\Lambda}(\P)\cap \cC(\P)$, with $\negbeta \neq \negbeta^{0,n}$. Since $\nabla_j^n(\negbeta)\neq \emptyset$ for any $j\in \{2,\ldots,n\}$, there exists an absolute maximal element $\negalpha^{i,n}\in \nabla_2^n(\negbeta)$, and thus (Theorem \ref{gamma hat}) $\beta_2=i$ and $\beta_3\geq i$. Similarly, there exists an absolute maximal element $\negalpha^{i',n}\in \nabla_3^n(\negbeta)$, and thus $\beta_3=i'$ and $\beta_2\geq i'$. Hence $i=i'$. Proceeding in the same way with pairs of the remaining indexes, we conclude that there exists an absolute maximal element $\negalpha^{i,n}\in \bigcap_{j=2}^n \nabla_j^n(\negbeta)$ and in particular that $\beta_j=i$ for $j=2,\ldots,n$. As $\negbeta\in \widehat{\Lambda}(\P)$ and $n\geq 3$, it follows that $\negbeta\neq \negalpha^{i,n}$ and thus $\beta_1>\alpha^{i,n}_1$. Hence, for each $\negbeta\in \widehat{\Lambda}(\P)\cap \cC(\P)$, there exists a unique $\negalpha^{i,n}\in~\widehat{\Gamma}(\P)\cap \cC(\P)$ such that $\negalpha^{i,n}\in~\nabla_{J}(\negbeta)$, where $J=\{2,\ldots,n\}$. Therefore, $\#(\widehat{\Gamma}(\P)\cap \cC(\P))\geq \#((\widehat{\Lambda}(\P)\cap \cC(\P))\setminus \{\negbeta^{0,n}\})$. As $\#R_n=\#(\widehat{\Gamma}(\P) \cap \cC(\P))$ and $R_n\subseteq ((\widehat{\Lambda}(\P)\cap \cC(\P))\setminus \{\negbeta^{0,n}\})$, we have $ ((\widehat{\Lambda}(\P)\cap \cC(\P))\setminus \{\negbeta^{0,n}\})=R_n$, which proves the result.
\end{proof}

\begin{corollary}\label{corollary-LambdaKummer}
Let $F/K(x)$ be the Kummer extension given in (\ref{KummerEquation}) and assume $g \geq 1$. For $2 \leq n \leq r-\floor*{\frac{r}{m}}$,
\begin{align*}
\Lambda^*(\P)=\Bigg\{(k_1m+i, k_2m+i, \dots, k_nm+i)&: 1\leq i \leq m-1-\floor*{\frac{m}{r}},\, k_1, k_2, \dots, k_n\geq 0,\\
&\quad k_1+\dots+k_n=r-2-\floor*{\frac{ri}{m}}\Bigg\}.
\end{align*}
\end{corollary}
\begin{proof}
It follows from Proposition \ref{Gammasymmetry} and Theorem \ref{relativemaxKummer}.
\end{proof}

\subsection{The pure gap set $G_0(\P)$}
We begin this section by presenting an explicit description for the set $G_{k, 0, \dots, 0}$ for $k\geq 0$.

\begin{proposition}\label{Kummer-Gk0..0}
For $0\leq k \leq r-n-1-\floor*{\frac{r}{m}}$, let 
\begin{equation} \label{Bk}
B_k:=\left\{(i_1, i_2, \dots, i_n)\in \N^n:  1\leq i_j \leq m-\ceil*{\frac{m(k+j)}{r}} \text{ for }j=1,\dots, n\right\}.
\end{equation}
Then
$$
G_{k, 0, \dots, 0}=(km, 0, \dots, 0)+\bigcup_{\si\in S_n}\si(B_k)
$$
for $0\leq k \leq r-n-1-\floor*{\frac{r}{m}}$, and $G_{k, 0, \dots, 0}=\emptyset$ for $r-n-\floor*{\frac{r}{m}}\leq k$.
\end{proposition}
\begin{proof}
We begin by remembering that an element $(km+i_1, i_2, \dots, i_n)$ is in $G_{k, 0, \dots, 0}$ if and only if there exist $\negbeta^1, \dots, \negbeta^n$ in $\Lambda^*(\P)$ such that $\glb(\negbeta^1, \dots, \negbeta^n)=(km+i_1, i_2, \dots, i_n)$ and
\begin{equation}\label{equation_condition_1}
\beta_j^j<\min\{\beta_j^1, \dots, \widehat{\beta_j^j}, \dots, \beta_j^n\}\quad \text{for}\quad 1\leq j \leq n.
\end{equation}
Furthermore, note that from (\ref{Gk1...km}) and Corollary \ref{corollary-LambdaKummer}, each element $\negbeta^j$ must be of the form
\begin{equation}\label{betas}
\negbeta^j=((k+k_1^j)m+i_j, k_2^jm+i_j, \dots, k_n^jm+i_j),
\end{equation}
where $k_j^j=0$, $k_s^j\geq 0$ for $s\neq j$, $1\leq i_j\leq m-1-\floor*{\frac{m}{r}}$, and 
\begin{equation}\label{equation_condition_2}
k+k^j=r-2-\floor*{\frac{ri_j}{m}},\quad \text{where}\quad k^j:=\sum_{s=1}^n k_s^j.
\end{equation}
Now, we divide the proof into three parts.\\

{\bf Claim 1:} $G_{k, 0, \dots, 0}=\emptyset$ for $r-n-\floor*{\frac{r}{m}}\leq k$.

{\it Proof of the Claim 1: }Suppose that there exists an element $(km+i_1, i_2, \dots, i_n)$ in $G_{k, 0, \dots, 0}$. From Corollary \ref{corollary-permutation-Gk0..0} we can assume, without loss of generality, that $i_1\geq i_2\geq \cdots \geq i_n$. Then there exists $\negbeta^1, \dots, \negbeta^n$ elements in $\Lambda^*(\P)$ as in (\ref{betas}) satisfying (\ref{equation_condition_1}) and (\ref{equation_condition_2}). 

From (\ref{equation_condition_1}) and since $i_1\geq i_2 \geq \cdots \geq i_n$, we obtain $k_s^j\geq 1$ for $1\leq s\leq n-1$ and $s+1\leq j\leq n$. This implies 
\begin{equation}\label{lowerbound_k^j}
j-1\leq k^j\quad \text{for}\quad 1\leq j \leq n.
\end{equation}

From (\ref{equation_condition_2}), for each $1\leq j \leq n$ we have 
$$
r-n-\floor*{\frac{r}{m}}+k^j\leq k+k^j=r-2-\floor*{\frac{ri_j}{m}}\leq r-2-\floor*{\frac{r}{m}}.
$$
Hence $0\leq k^j\leq n-2$ for $1\leq j\leq n$. In particular, for $j=n$ we obtain $k^n\leq n-2$, a contradiction with (\ref{lowerbound_k^j}). Therefore $G_{k, 0, \dots, 0}=\emptyset$.\\

{\bf Claim 2: }$G_{k, 0, \dots, 0}\subseteq (km, 0, \dots, 0)+\cup_{\sigma\in S_n}\sigma(B_k)$.

{\it Proof of the Claim 2: }Let $(km+i_1, i_2, \dots, i_n)\in G_{k, 0, \dots, 0}$. Analogously to the proof of Claim 1, we can assume that $i_1\geq i_2\geq \cdots \geq i_n$, and therefore there exists $\negbeta^1, \dots, \negbeta^n$ elements in $\Lambda^*(\P)$ as in (\ref{betas}) satisfying (\ref{equation_condition_1}) and (\ref{equation_condition_2}), and $j-1\leq k^j$ for $1\leq j \leq n$. This implies $k+j\leq k+k^j + 1$. Now, from (\ref{equation_condition_2}), we get $k+j \leq r-1-\floor*{\frac{ri_j}{m}} = r- \ceil*{\frac{ri_j}{m}}$, and thus $\frac{m(k+j)}{r} \leq m - \frac{m}{r}  \ceil*{\frac{r i_j}{m}}$. So, we obtain 
$$
1\leq i_j \leq m-\ceil*{\frac{m(k+j)}{r}}\quad \text{for}\quad 1\leq j \leq n.
$$
Therefore $(km+i_1, i_2, \dots, i_n)\in (km, 0, \dots, 0)+B_k$.\\

{\bf Claim 3: }$(km, 0, \dots, 0)+\cup_{\sigma\in S_n}\sigma(B_k)\subseteq G_{k, 0, \dots, 0}$.

{\it Proof of the Claim 3: }Let $(i_1, i_2, \dots, i_n)$ be an element in $\cup_{\sigma\in S_n}\sigma(B_k)$. Then there exists a permutation $\tau\in S_n$ such that
\begin{equation}\label{i_1>i_2>...}
i_{\tau(1)}\geq i_{\tau(2)}\geq \cdots \geq i_{\tau(n)}.
\end{equation}
We are going to prove that $(km+i_{\tau(1)}, i_{\tau(2)}, \dots, i_{\tau(n)})\in G_{k, 0, \dots, 0}$. For this, we need to construct elements $\negalpha^1, \dots, \negalpha^n$ in $\Lambda^*(\P)$ such that
$\alpha_j^j<\min\{\alpha_j^1, \dots, \widehat{\alpha_j^j}, \dots, \alpha_j^n\}$ for $1\leq j \leq n$ and $\glb(\negalpha^1, \dots, \negalpha^n)=(km+i_{\tau(1)}, i_{\tau(2)}, \dots, i_{\tau(n)})$.
 
First, note that the inequality (\ref{i_1>i_2>...}) implies that $(i_{\tau(1)}, i_{\tau(2)}, \dots, i_{\tau(n)})\in B_k$, that is, 
$$
1\leq i_{\tau(j)}\leq m-\ceil*{\frac{m(k+j)}{m}}\quad \text{for}\quad 1\leq j \leq n.
$$
Now, for each $j$ such that $1\leq i_{\tau(j)} \leq m-\ceil*{\frac{m(k+n)}{r}}$, define $\tilde{k}^j:=r-k-2-\floor*{\frac{ri_{\tau(j)}}{m}}$. Then
$$
n-1\leq \tilde{k}^j \leq r-k-2-\floor*{\frac{r}{m}}.
$$
For this case, define $\tilde{k}_j^j:=0$ and let $\tilde{k}_1^j, \dots, \tilde{k}_{j-1}^j, \tilde{k}_{j+1}^j, \dots, \tilde{k}_n^j\in \left\{1, 2, \dots, r-2-\floor*{\frac{r}{m}}\right\}$ be integers such that $\tilde{k}^j=\sum_{s=1}^{n}\tilde{k}_s^j$. 

On the other hand, for each $j$ such that $m-\floor*{\frac{m(k+n)}{r}}\leq i_{\tau(j)}$, there exist an unique integer $\tilde{k}^j$ in $\{j-1, \dots , n-2\}$ such that 
$$
m-\floor*{\frac{m(k+\tilde{k}^j+2)}{r}}\leq i_{\tau(j)}\leq m-\ceil*{\frac{m(k+\tilde{k}^j+1)}{r}}.
$$
This implies that $k+\tilde{k}^j=r-2-\floor*{\frac{ri_{\tau(j)}}{m}}$. Furthermore, for $1\leq s\leq n$ define
$$ 
\tilde{k}^j_s= \left\{
\begin{array}{ll}
1, & \text{if }s\in \{1, 2, \dots, \dots, \tilde{k}^j+1\}\setminus \{j\},\\
0,  & \text{in otherwhise}.
\end{array} \right. 
$$
Thus, we have that $\tilde{k}^j=\sum_{s=1}^n \tilde{k}_s^j$. 

Now, for $1\leq j \leq n$ define 
$$
\negalpha^j:=((k+\tilde{k}_1^j)m+i_{\tau(j)}, \tilde{k}_2^jm+i_{\tau(j)}, \dots, \tilde{k}_n^jm+i_{\tau(j)}).
$$
Since $k+\tilde{k}^j=r-2-\floor*{\frac{ri_{\tau(j)}}{m}}$ for each $j$, we obtain $\negalpha^j\in \Lambda^*(\P)$. Furthermore, from the construction of the integers $\tilde{k}_1^j, \tilde{k}_2^j, \dots, \tilde{k}_n^j $ and from (\ref{i_1>i_2>...}), we obtain $\alpha_j^j<\min\{\alpha_j^1, \dots, \widehat{\alpha_j^j}, \dots, \alpha_j^n\}$ for $1\leq j \leq n$ and $\glb(\negalpha^1, \dots, \negalpha^n)=(km+i_{\tau(1)}, i_{\tau(2)}, \dots, i_{\tau(n)})$. This implies that $(km+i_{\tau(1)}, i_{\tau(2)}, \dots, i_{\tau(n)})\in G_{k, 0, \dots, 0}$. Finally, from Corollary \ref{corollary-permutation-Gk0..0}, we have $(km+i_1, i_2, \dots, i_n)\in G_{k, 0, \dots, 0}$.
\end{proof}

In order to calculate the cardinality of the set $G_{k, 0, \dots, 0}$, we define a family of functions inductively. For $n=1$ define the function $D_1:\Z\rightarrow \Z$ such that $D_1(a_1)=a_1$ for all $a_1\in \Z$. For $n\geq 2$ define the function $D_n:\Z^n\rightarrow \Z$ given by
\begin{equation}\label{function-Sn}
D_n(a_1, \dots, a_n)=a_n^n+\sum_{i=1}^{n-1}\binom{n}{i}a_n^{n-i}D_i(a_1-a_n,\dots, a_i-a_n).
\end{equation}

\begin{lemma}\label{properties_Sn}
Let $D_n$ be the function defined in (\ref{function-Sn}). The following statements hold:
\begin{enumerate}[\rm (i)]
\item $D_n(ua_1, \dots, ua_n)=u^nD_n(a_1, \dots, a_n)$ for $u\in \Z$.
\item If $a_1,  \dots, a_n$ is a non-increasing consecutive sequence of integers, that is, $a_i=a_n+n-i$ for $i=1, \dots, n$, then
$$
D_n(a_1,  \dots, a_n)=a_n(a_n+n)^{n-1}.
$$
\end{enumerate}
\end{lemma}
\begin{proof}
The proof is by induction on $n$. Observe that both properties are true for $n=1$. Now, assume that the properties are true for any positive integer less than $n$. Then, for the first item,
\begin{align*}
D_n(ua_1,  \dots, ua_n)&=(ua_n)^n+\sum_{i=1}^{n-1}\binom{n}{i}(ua_n)^{n-i}D_i(ua_1-ua_n,\dots, ua_i-ua_n)\\
&=u^na_n^n+\sum_{i=1}^{n-1}\binom{n}{i}u^{n-i}a_n^{n-i}D_i(u(a_1-a_n),\dots, u(a_i-a_n))\\
&=u^na_n^n+\sum_{i=1}^{n-1}\binom{n}{i}u^na_n^{n-i}D_i(a_1-a_n,\dots, a_i-a_n)\\
%&=u^n\left(a_n^n+\sum_{i=1}^{n-1}\binom{n}{i}a_n^{n-i}D_i(a_1-a_n,\dots, a_i-a_n)\right)\\
&=u^nD_n(a_1, \dots, a_n).
\end{align*}
On the other hand, for the second item,
\begin{align*}
D_n(a_1, \dots, a_n)&=a_n^n+\sum_{i=1}^{n-1}\binom{n}{i}a_n^{n-i}D_i(a_1-a_n,\dots, a_i-a_n)\\
&=a_n\left(a_n^{n-1}+\sum_{i=1}^{n-1}\binom{n}{i}a_n^{n-1-i}D_i(n-1, \dots, n-i)\right)\\
&=a_n\left(a_n^{n-1}+\sum_{i=1}^{n-1}\binom{n}{i}a_n^{n-1-i}(n-i)n^{i-1}\right)\\
&=a_n\left(a_n^{n-1}+\sum_{i=1}^{n-1}\binom{n-1}{i}a_n^{n-1-i}n^i\right)\\
&=a_n(a_n+n)^{n-1}.
\end{align*}
\end{proof}

The following result establishes that the function $D_n$ given in (\ref{function-Sn}) determines the cardinality of the set $\cup_{\si\in S_n}\si(B_k)$ for each $k$, where $B_k$ is the set defined in (\ref{Bk}). The proof will be omitted because the argument is exactly the same as the one given in \cite[Lemma 5]{YH2018}.

\begin{lemma}\label{Cardinality-permutation}
Let $a_1, \dots, a_n$ be positive integers such that $a_1\geq a_2\geq \dots \geq a_n$. Define the set 
$$
B:=\{ \negi =(i_1, \dots, i_n)\in \N^n: 1\leq i_j \leq a_j \text{ for }j=1, \dots, n\}.
$$ 
Then
$$
\left|\bigcup_{\si\in S_{n}}\si(B)\right|=D_n(a_1,  \dots, a_n).
$$
\end{lemma}

\begin{theorem}\label{Kummer-PureGaps} 
The pure gap set $G_0(\P)$ is given by
%$$
%G_0(\P)=\bigcup_{\substack{0\leq k_1, k_2, \dots, k_n\\k_1+\cdots+k_n=k\\ 0\leq k \leq r-n-1-\floor*{r/m}}}\left((k_1m, k_2m, \dots, k_nm)+\cup_{\sigma\in S_I}\sigma(B_k)\right),
%$$
\begin{align*}
G_0(\P)=\Bigg\{(k_1m, \dots, k_nm)+\si(\negi)&:\, k_1, \dots, k_n\geq 0, \, k_1+\cdots +k_n=k,\\
&\quad 0\leq k \leq r-n-1-\floor*{\frac{r}{m}},\, \negi\in B_{k},  \text{ and }\si\in S_n\Bigg\},
\end{align*}
where $B_k$ is the set defined in Proposition \ref{Kummer-Gk0..0}. Furthermore,
$$
|G_0(\P)|=\sum_{k=0}^{r-n-1-\floor*{r/m}}\binom{k+n-1}{n-1}D_n\left(m-\ceil*{\frac{m(k+1)}{r}}, \dots, m-\ceil*{\frac{m(k+n)}{r}}\right).
$$
\end{theorem}
\begin{proof}
To obtain such description of the pure gap set $G_0(\P)$, it is enough to observe that
\begin{align*}
G_0(\P)&=\bigcup_{0\leq k_1,  \dots, k_n}G_{k_1, \dots, k_n}& \text{(from Equation (\ref{PureGapsDecomposition}))}\\
&=\bigcup_{\substack{0\leq k_1,  \dots, k_n\\k_1+\cdots+k_n=k\\ 0\leq k \leq r-n-1-\floor*{r/m}}}(G_{k, 0, \dots, 0}+\w_{k_2, \dots, k_n})& \text{(from Proposition \ref{Puregapsymmetry})}\\
&=\bigcup_{\substack{0\leq k_1,  \dots, k_n\\k_1+\cdots+k_n=k\\ 0\leq k \leq r-n-1-\floor*{r/m}}}((k_1m,  \dots, k_nm)+\cup_{\sigma\in S_n}\sigma(B_k)).& \text{(from Proposition \ref{Kummer-Gk0..0})}
\end{align*}
The formula for the cardinality of the set $G_0(\P)$ follows from (\ref{CardinalityPuregaps}), Proposition \ref{Kummer-Gk0..0}, and Lemma \ref{Cardinality-permutation}.
\end{proof}

\begin{example} \label{ex numerico}
Consider the Kummer extension defined over $K$ by the affine equation $\mathcal{X}: y^5 = f(x)$, where $f(x)$ is a separable polynomial of degree $r=9$ and $\char(K)\nmid 5$. The genus of the function field $F = K(\mathcal{X})$ is $g=16$. Let $\P=(P_1, P_2, P_3)$, where $P_1, P_2$ and $P_3$ are totally ramified places on the extension $F/K(x)$ distinct to $P_\infty$. In Figure \ref{image} we present a diagram with the elements in the sets $G_0(\P)$ and $\Lambda^*(\P)$. The cubes with side equal to the period $m=5$ represent the following sets:

$\bullet$ one set $G_{0,0,0}$ with $54$ pure gaps: one pink cube;

$\bullet$ three sets of the form $G_{1,0,0} + \w_{k_2, k_3}$ with $26$ pure gaps in each of these sets: three purple cubes;

$\bullet$ six sets of the form $G_{2,0,0} + \w_{k_2, k_3}$ with $20$ pure gaps in each of these sets: six green cubes;

$\bullet$ ten sets of the form $G_{3,0,0} + \w_{k_2, k_3}$ with $7$ pure gaps in each of these sets: ten blue cubes;

$\bullet$ fifteen sets of the form $G_{4,0,0} + \w_{k_2, k_3}$ with $4$ pure gaps in each of these sets: fifteen yellow cubes.

In \url{https://www.geogebra.org/m/rmrw2y7j} we provide a three-dimensional interactive image of Figure \ref{image} where we can appreciate the symmetry of the set $G_0(\P)$ and the relationship between the sets $G_0(\P)$ and $\Lambda^*(\P)$ given by Corollary \ref{coro_G0(Q)}.

From Lemma \ref{Cardinality-permutation} and Theorem \ref{Kummer-PureGaps} we have $|G_0(\P)| = 54 + 3 \cdot 26 + 6 \cdot 20 + 10\cdot 7 + 15 \cdot 4 = 382$. 
\end{example}

%From Theorem \ref{Kummer-PureGaps}, we can see that the set of pure gaps $G_0(\P)$ is invariant under permutations, that is, $\si(G_0(\P))=G_0(\P)$ for any $\si\in S_n$, verifying what is established in Proposition \ref{Gapspermutation}. As an example, in Figure \ref{image} we show the set of pure gaps $G_0(\P)$ and the set of maximal relative elements $\Lambda^*(\P)$ for the case $m=5$, $r=9$ and $n=3$. In addition, in \url{https://www.geogebra.org/m/rmrw2y7j} we provide a three-dimensional interactive image of Figure \ref{image} where we can appreciate the symmetry of the set $G_0(\P)$ and the relationship between the sets $G_0(\P)$ and $\Lambda^*(\P)$ given by Corollary \ref{coro_G0(Q)}.\\

\begin{figure}
\begin{center}
\includegraphics[width=0.7\textwidth]{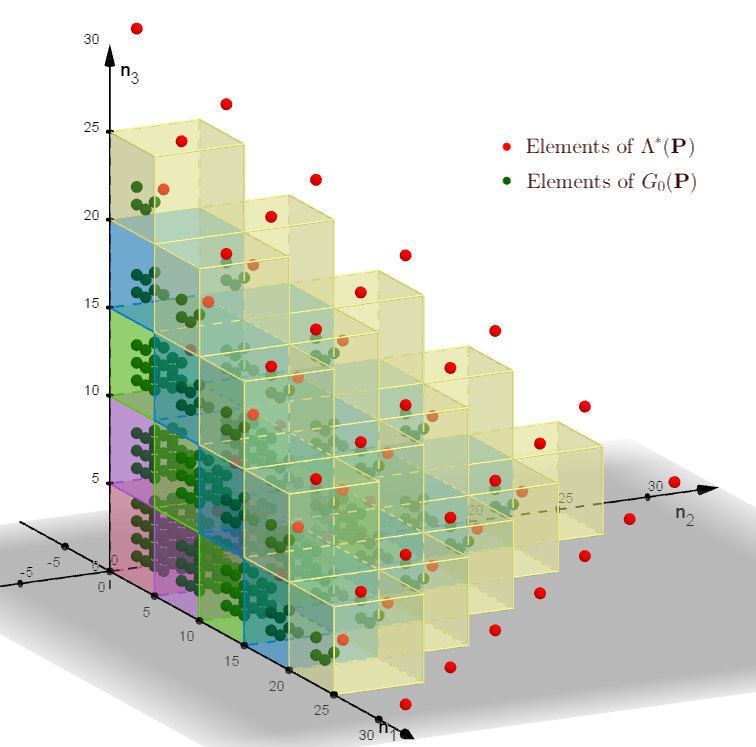}  
\caption{The sets $G_0(\P)$ and $\Lambda^*(\P)$ for the case $m=5$, $r=9$, and $n=3$.}
\label{image}
\end{center}
\end{figure}

In recent works the cardinality of the pure gap set $G_0(\P)$ was determined in some particular cases. For example, in \cite[Theorema 4.12]{CMT2024} the authors determined the cardinality of $G_0(\P)$ for $n=2$, that is when $\P=(P_1, P_2)$, and in \cite[Theorem 6]{YH2018} the authors determined the cardinality of $G_0(\P)$ over the Hermitian function field, that is, when $m=q+1$ and $r=q$. Next, we show that the formula for the cardinality of $G_0(\P)$ given in Theorem \ref{Kummer-PureGaps} generalizes these results.\\

$\bullet$ Case $n=2$. Note that the function $D_2$ is given by 
\begin{align*}
D_2(a_1, a_2)&=a_2^2+2a_2(a_1-a_2)\\
&=a_2^2+2a_2(a_1-a_2)+(a_1-a_2)^2-(a_1-a_2)^2\\
&=a_1^2-(a_1-a_2)^2.
\end{align*}
Therefore, from Theorem \ref{Kummer-PureGaps}, we obtain
\begin{align*}
|G_0(P_1, P_2)|&=\sum_{k=0}^{r-3-\floor*{\frac{r}{m}}}\binom{k+1}{1}D_2\left(m-\ceil*{\frac{m(k+1)}{r}}, m-\ceil*{\frac{m(k+2)}{r}}\right)\\
&=\sum_{k=0}^{r-3-\floor*{\frac{r}{m}}}(k+1)\left[\left(m-\ceil*{\frac{m(k+1)}{r}}\right)^2-\left(\ceil*{\frac{m(k+2)}{r}}-\ceil*{\frac{m(k+1)}{r}}\right)^2\right]\\
&=\sum_{k=1}^{r-2-\floor*{\frac{r}{m}}}k\left[\left(m-\ceil*{\frac{mk}{r}}\right)^2-\left(\ceil*{\frac{m(k+1)}{r}}-\ceil*{\frac{mk}{r}}\right)^2\right].
\end{align*}
This formula coincides with the one given in \cite[Theorem 4.12]{CMT2024}.\\

$\bullet$ Case $m=ur+1$ with $u\in \N$. In this case we have
\begin{align*}
|G_0(\P)|&=\sum_{k=0}^{r-n-1}\binom{k+n-1}{n-1}D_n\left(u(r-k-1), \dots, u(r-k-n)\right)\\
&=u^n\sum_{k=0}^{r-n-1}\binom{k+n-1}{n-1}D_n\left(r-k-1, \dots, r-k-n\right)& \text{(from Lemma \ref{properties_Sn} (i))}\\
&=u^n\sum_{k=0}^{r-n-1}\binom{k+n-1}{n-1}(r-k-n)(r-k)^{n-1}.& \text{(from Lemma \ref{properties_Sn} (ii))}
\end{align*}
In particular, for $u=1$ and $r=q$ we obtain the formula for the cardinality of the pure gap set $G_0(\P)$ over the Hermitian function field which coincides with the formula given in \cite[Theorem 6]{YH2018}.

\subsection{AG Codes in many places over Kummer Extensions}

Finally, in this subsection we will use the simple explicit description of the pure gap set given in Theorem \ref{Kummer-PureGaps} to provide a general construction of differential AG codes in many places over the Kummer extension $\fq(\cX)/\fq(x)$ defined by the equation given in (\ref{KummerEquation}). For this, our main tool will be Theorem \ref{TorresCarvalho}.

Let $2\leq n \leq r-1-\floor*{\frac{r}{m}}$. From Theorem \ref{Kummer-PureGaps}, the pure gap set $G_0(\P)$ can be decomposed as the disjoint union
$$
G_0(\P)=\bigcup_{\substack{0\leq k_1, \dots, k_n\\k_1+\cdots+k_n=k\\ 0\leq k \leq r-n-1-\floor*{r/m}}}G_{k_1, k_2, \dots, k_n},
$$ 
where $G_{k_1, \dots, k_n}=\{(k_1m, \dots, k_nm)+\si(\negi): \negi\in B_k \text{ and }\si\in S_n\}$. Therefore we can use each set $G_{k_1, \dots, k_n}$ of pure gaps at $\P$ to construct codes using Theorem \ref{TorresCarvalho}. For $\si\in S_n$, we define the following subset of $G_{k_1, \dots, k_n}$:
$$
G_{k_1, \dots, k_n}^\si:=\{(k_1m, \dots, k_nm)+\si(\negi): \negi\in B_k\}.
$$

\begin{remark} \label{G sigma}
Let $\negalpha, \negbeta \in G_{k_1, \dots, k_n}^\si$. If $\neggamma\in \N^n$ with $\negalpha \leq \neggamma \leq \negbeta$, then $\neggamma \in G_{k_1, \dots, k_n}^\si$. In fact, suppose that $\negalpha = (k_1m, \dots, k_nm) + \si(\negi)$ and $\negbeta= (k_1m, \dots, k_nm) + \si(\negi')$, with $\negi , \negi' \in B_k$. Since $\negalpha \leq \neggamma \leq \negbeta$, we have that $\neggamma = (k_1m, \dots, k_nm) + \mathbf{t}$, with $\si(\negi) \leq \mathbf{t} \leq \si(\negi')$. Then, from definition of $B_k$ we get $\widehat{\negi} \in B_k$ such that $\mathbf{t} = \sigma(\widehat{\negi})$ and so $\neggamma \in G_{k_1, \dots, k_n}^\si$.
\end{remark}

This fact implies that the sets $G_{k_1, \dots, k_n}^\si$ are suitable for construct codes using Theorem \ref{TorresCarvalho}.

Without loss of generality, in the following result we construct differential AG codes using the set $G_{k_1, \dots, k_n}^\si$ for $\si=\Id$. It is not difficult to see that using the sets $G_{k_1, \dots, k_n}^\si$ for $\si\neq \Id$ we obtain codes with the same parameters.

\begin{theorem}\label{teo_codes}
Let $\cX$ be the curve of genus $g$ given in  (\ref{KummerEquation}) and $P_1, \dots, P_n$ be totally ramified places distinct to $P_\infty$ in the extension $\fq(\cX)/\fq(x)$, where $2\leq n\leq r-1-\floor*{\frac{r}{m}}$. Let $k\in \left\{0, \dots, r-n-1-\floor*{\frac{r}{m}}\right\}$. For $(a_1, a_2, \dots, a_n)\in B_k$ and $k_1, \dots, k_n\geq 0$ be such that $k_1+\cdots + k_n=k$, define the divisors
$$
G_k:=\sum_{j=1}^{n}\left(2k_jm+a_j+m-1-\ceil*{\frac{m(k+j)}{r}}\right)P_j\quad \text{and}\quad D:=\sum_{\substack{P\in \cX(\fq)\\ P\not\in \{P_1, \dots, P_n\}}}P.
$$
If $2g-2<\deg (G_k)<N:=\deg (D)$, then the differential AG code $C_{\Omega}(D, G_k)$ has parameters $[N, k_\Omega, d_\Omega]$, where
\begin{align*}
&k_\Omega=N+n-2km-nm-1+\frac{(m-1)(r-1)}{2}-\sum_{j=1}^{n}a_j+\sum_{j=1}^{n}\ceil*{\frac{m(k+j)}{r}} \text{ and}\\
&d_\Omega \geq 2km+2nm-mr+m+r+1-2\sum_{j=1}^{n}\ceil*{\frac{m(k+j)}{r}}.
\end{align*}
\end{theorem}
\begin{proof}
For $k\in \left\{0, \dots, r-n-1-\floor*{\frac{r}{m}}\right\}$, define $
\negalpha:=(k_1m+a_1, \dots, k_nm+a_n)$ and
$$
\negbeta:=\left(k_1m+m-\ceil*{\frac{m(k+1)}{r}}, \dots, k_nm+m-\ceil*{\frac{m(k+n)}{r}}\right).
$$
So, $\negalpha$ and $\negbeta$ are elements in $G_{k_1, \dots, k_n}^{\Id}$ such that $\negalpha\leq \negbeta$. Let $[N, k_\Omega, d_\Omega]$ be the parameters of the code $C_{\Omega}(D, G_k)$. If $2g-2<\deg (G_k)<N$, then from Proposition \ref{goppabound} we obtain
\begin{align*}
k_\Omega&=N+g-1-\deg(G_k)\\
&=N+\frac{(m-1)(r-1)}{2}-1-\left(2km+nm-n+\sum_{j=1}^{n}a_j-\sum_{j=1}^{n}\ceil*{\frac{m(k+j)}{r}}\right)\\
&=N+n-2km-nm-1+\frac{(m-1)(r-1)}{2}-\sum_{j=1}^{n}a_j+\sum_{j=1}^{n}\ceil*{\frac{m(k+j)}{r}}.
\end{align*}
Now, by Remark \ref{G sigma}, we have that $\neggamma \in G_0(\mathbf{P})$ for each $\neggamma \in \N^n$ with $\negalpha \leq \neggamma \leq \negbeta$. Thus, from Theorem \ref{TorresCarvalho} we obtain
\begin{align*}
d_\Omega&\geq \deg(G_k)-(2g-2)+n+\sum_{j=1}^{n}\left(m-\ceil*{\frac{m(k+j)}{r}}-a_j\right)\\
&=2km+nm-n+\sum_{j=1}^{n}a_j-\sum_{j=1}^{n}\ceil*{\frac{m(k+j)}{r}}-(mr-m-r-1)+n\\
&\quad +mn-\sum_{j=1}^{n}\ceil*{\frac{m(k+j)}{r}}-\sum_{j=1}^{n}a_j\\
&=2km+2nm-mr+m+r+1-2\sum_{j=1}^{n}\ceil*{\frac{m(k+j)}{r}}.
\end{align*}
\end{proof}

Note that the dimension $k_\Omega$ and the lower bound for the minimum distance $d_\Omega$ obtained in the previous theorem depend on the value of $k$, but not on the values of $k_1, k_2, \dots, k_n$. This means that to obtain codes with the same parameters as those described in the previous theorem, it is enough to use the sets of pure gaps of the form $G_{k, 0, \dots, 0}$ for $k\geq 0$.

\begin{example}\label{Example1}
Consider the subcover of the Hermitian curve given by the affine equation $y^m=x^q+x$, where $m$ is a divisor of $q+1$. This curve is maximal over $\fqs$ and has $q+1+m(q^2-q)$ rational points over $\fqs$. Using the same notation of Theorem \ref{teo_codes}, let $a:=\sum_{j=1}^{n}a_j$. In Table \ref{table1} we present the parameters of the differential AG codes obtained using Theorem \ref{teo_codes}.

\begin{table}[h!]
	\centering
	\begin{tabular}{|c|c|c|c|c|c|}\hline
		\, $q$\, & \, $m$\,  & \, $n$\,  & \, $k$ \, & \, $a$\,  & \, $[N, k_\Omega, d_\Omega]$ \\ \hline  
		$7$  & $4$ &  $2$ & $3$ & $2$ &   $[174, 156, \geq 12]$ \\     \hline
		$8$  & $3$ &  $2$ & $3$ & $2$ &   $[ 175 , 161 ,  \geq 10 ]$ \\    \hline
		$9$ & $5$ & $2$ & $5$ & $2$ & $[ 368 , 331 ,  \geq 24 ]$  \\ \hline
		$9$  & $5$ &  $3$ & $4$ & $3, 4$ &   $[367, 341-a, \geq 18]$ \\    \hline
	\end{tabular}
	\captionof{table}{Parameters of AG codes over the curve $y^m=x^q+x$.}
	\label{table1}
\end{table}
\end{example}

\begin{example}\label{Example2}
Consider the curve defined by $y^m=(x^{q^{t/2}}-x)^{q^{t/2}-1}$ over $\mathbb{F}_{q^t}$, where $t$ is even, $m$ is a divisor of $q^t-1$, and $\gcd(m, q^{t/2}-1)=1$. This curve has $(q^t-q^{t/2})m+q^{t/2}+1$ rational points over $\mathbb{F}_{q^t}$. Similarly to the previous example, from Theorem \ref{teo_codes} we obtain families of differential AG codes. The parameters of these codes are presented in Table \ref{table2}.

\begin{table}[h!]\label{table2}
	\centering
	\begin{tabular}{|c|c|c|c|c|c|c|}\hline
		\, $q$\, & \, $t$\, & \, $m$\,  & \, $n$\,  & \, $k$ \, & \, $a$\,  & \, $[N, k_\Omega, d_\Omega]$ \\ \hline 
		\, $2$\, & \, $6$\, & \, $3$\,  & \, $2$\,  & \, $3$ \, & \, $2$\,  & \, $[175, 161, \geq 10]$ \\ \hline 
		\, $2$\, & \, $6$\, & \, $9$\,  & \, $2$\,  & \, $5$ \, & \, $2, 3$\,  & \, $[511, 447-a, \geq 42]$ \\ \hline 
		\, $2$\, & \, $6$\, & \, $9$\,  & \, $3$\,  & \, $4$ \, & \, $3, 4, 5$\,  & \, $[510, 462-a, \geq 30]$ \\ \hline  
		\, $3$\, & \, $4$\, & \, $5$\,  & \, $3$\,  & \, $4$ \, & \, $3, 4$\,  & \, $[367, 341-a, \geq 18]$ \\ \hline  
		\end{tabular}
	\captionof{table}{Parameters of AG codes over the curve $y^m=(x^{q^{t/2}}-x)^{q^{t/2}-1}$.}
	\label{table2}
\end{table}
\end{example}

The codes obtained in the previous examples (see Tables \ref{table1} and \ref{table2}) improve the parameters given in MinT’s Tables \cite{MinT}. %Furthermore, all  codes over $\mathbb{F}_{64}$ obtained in Example \ref{Example2} has better relative parameters than the AG codes obtained in \cite[Example 3]{HY2018}.

We observe that, by \cite[Exercise 7 (iii)]{TV1991}, if there is a linear code over $\mathbb{F}_{q}$ of length $n$, dimension $k$ and minimum distance $d$, then for each nonnegative integer $s < k$, there exists a linear code over $\mathbb{F}_{q}$ of length $n-s$, dimension $k-s$, and minimum distance $d$. 

Thus, from the codes obtained in Example \ref{Example1}, we obtain new linear codes with parameters $[367-s, 338-s, \geq 18]$ with $0\leq s\leq 7$ over $\mathbb{F}_{81}$, which have better relative parameters $R+\delta$ than the AG codes of length $367-s$ obtained in \cite[Example 1]{HY2018}.

Analogously, from the codes obtained in Example \ref{Example2}, we obtain new linear codes with parameters $[510-s, 489-s, \geq 30]$ with $0\leq s\leq 110$ over $\mathbb{F}_{64}$, which have better relative parameters $R+\delta$ than the AG codes of length $510-s$ obtained in \cite[Example 3]{HY2018}.

%STYLE:
%\bibliographystyle{alpha}
%\bibliographystyle{amsalpha}
\bibliographystyle{abbrv}
%\bibliographystyle{acm}
%\bibliographystyle{apalike}
%\bibliographystyle{ieeetr}
%\bibliographystyle{plain}
%\bibliographystyle{siam}
%\bibliographystyle{unsrt}

%\nocite{*}
\bibliography{puregapsII} 

\end{document}